\documentclass[10pt,twocolumn,twoside]{IEEEtran}

\usepackage{cite}

\ifCLASSINFOpdf
\else
\fi
\usepackage{amsmath}
\usepackage{graphicx} 
\usepackage{amsthm}
\usepackage{algorithm}  
\usepackage{algorithmic}
\usepackage{algorithm}
\usepackage{epstopdf}
\usepackage{latexsym,bm,amsmath, amssymb,amsfonts}%
\usepackage{enumitem}
\usepackage{booktabs} 
\usepackage{color}
\usepackage[format=plain]{caption}
\usepackage{subfigure}
\newtheorem{thm}{Theorem}
\newtheorem{assp}{Assumption}

\newtheorem{prop}{Proposition}

\newtheorem{rmk}{Remark}

\theoremstyle{plain}

\definecolor{orange}{RGB}{255,107,0}

\def\black{\color{black}}

\begin{document}
%
\title{Efficient Estimation of Sensor Biases for the 3-Dimensional Asynchronous Multi-Sensor System}
%
%
%

\author{Wenqiang Pu,~\IEEEmembership{Member,~IEEE,} Ya-Feng Liu,~\IEEEmembership{Senior Member,~IEEE,} Zhi-Quan Luo,~\IEEEmembership{Fellow,~IEEE}
\thanks{Part of this work has been presented at the 2018 IEEE GlobalSIP~\cite{Jiang2018}. The work of W. Pu  was supported by the National Natural Science Foundation of China (NSFC) under Grant 62101350. The work of Y.-F. Liu was supported in part by the NSFC under Grant 12022116 and Grant 12288201. The work of Z.-Q. Luo was supported in part by the National Key Research and Development Program of China Project under Grant 2022YFA1003900 and the Guangdong Provincial Key Laboratory of Big Data Computing. (\textit{Corresponding Author: Ya-Feng Liu.})}
\thanks{W. Pu and Z.-Q.~Luo are with Shenzhen Research Institute of Big Data, The Chinese University of Hong Kong, Shenzhen, 518172, China (e-mail: $\{$wenqiangpu,luozq$\}$@cuhk.edu.cn).}
\thanks{Y.-F.~Liu is with the State Key Laboratory
	of Scientific and Engineering Computing, Institute of Computational
	Mathematics and Scientific/Engineering Computing, Academy of
	Mathematics and Systems Science, Chinese Academy of Sciences,
	Beijing 100190, China (e-mail:
	yafliu@lsec.cc.ac.cn).}
}

\maketitle

\begin{abstract}
	
	An important preliminary procedure in multi-sensor data fusion is \textit{sensor registration}, and the key step in this procedure is to estimate sensor biases from their noisy measurements. There are generally two difficulties in this bias estimation problem: one is the unknown target states which serve as the nuisance variables in the estimation problem, and the other is the highly nonlinear coordinate transformation between the local and global coordinate systems of the sensors. In this paper, we focus on the 3-dimensional asynchronous multi-sensor scenario and propose a weighted nonlinear least squares (NLS) formulation by assuming that there is a target moving with a nearly constant velocity. We propose two possible choices of the weighting matrix in the NLS formulation, which correspond to classical and weighted NLS estimation and maximum likelihood (ML) estimation, respectively. To address the intrinsic nonlinearity, we propose a block coordinate descent (BCD) algorithm for solving the formulated problem, which alternately updates different kinds of bias estimates. Specifically, the proposed BCD algorithm involves solving linear LS problems and nonconvex quadratically constrained quadratic program (QCQP) problems with special structures. Instead of adopting the semidefinite relaxation technique, we develop a much more computationally efficient algorithm based on the alternating direction method of multipliers (ADMM) to solve the nonconvex QCQP subproblems. The convergence of the ADMM to the global solution of the QCQP subproblems is established under mild conditions. The effectiveness and efficiency of the proposed BCD algorithm are demonstrated via numerical simulations.
\end{abstract}


%
\IEEEpeerreviewmaketitle

\section{Introduction}
\label{sec:intro}
In the past decades, multi-sensor data fusion has attracted a lot of research interests in various applications, e.g., tracking system~\cite{bar1990multitarget,yan2020optimal}. Compared with the single sensor system, the performance can be significantly improved by integrating inexpensive stand-alone sensors into an integrated multi-sensor system, which is also a cost-effective way in practice. However, the success of data fusion not only depends on the data fusion algorithms but also relies on an important calibration process called \textit{sensor registration}. The sensor registration process refers to expressing each sensor's local data in a common reference frame, by removing the biases caused by the improper alignment of each sensor \cite{fischer1980registration}. Consequently, the key step in the sensor registration process is to estimate the sensor biases. 

\subsection{Related Works}
An early work \cite{fischer1980registration} for sensor registration dates back to the 1980s, where the authors identified various factors which dominate the alignment errors in the multi-sensor system and established a bias model to compensate the  alignment errors. Under the assumption that there exists a bias-free sensor, the maximum likelihood (ML) estimation approach \cite{fischer1980registration,Dana1990Registration,Fortunati2013On} and the least squares (LS) estimation approach \cite{Fortunati2011Least} were proposed. However, the bias-free sensor assumption is not practical and a more common situation in practice is that all sensors contain biases. For the general situation where all sensors are biased, various approaches have been proposed in the last four decades. From the perspective of parameter estimation, these approaches can be divided into three types: LS estimation \cite{leung1994least,Cowley1993Registration}, ML estimation \cite{zhou1997an,Ristic2003Sensor,Fortunati2013On}, Bayesian estimation \cite{helmick1993removal,Nabaa1999Solution,Okello2004Joint,Lin2006Multisensor,Li2010Joint,Taghavi2016practical,Zia2008EM,crouse2009sensor,mahler2011bayesian,ristic2013calibration}. Note that the sensor registration problem not only contains unknown sensor biases to be estimated, but also the target states, i.e., positions of the target at different time instances. The major differences among the above three types of approaches lie in the treatment for the unknown target states. For the LS approach, the target states are expressed as nonlinear functions of sensor biases in a common coordinate system, where the target states are eliminated and only sensor biases need to be estimated. As for the ML approach, a joint ML function for the sensor biases and the target states is formulated and maximized in an iterative manner (with two steps per iteration), i.e., one step to estimate the sensor biases and one step to estimate the target's states. Instead of jointly estimating the sensor biases and the target states as in the ML approach, the (recursive) Bayesian approach exploits the Markov property of the transition of the target states between different time instances, i.e., the probability distribution of the target state at the next time instance only depends on the target state at the current time instance. The sensor biases and the target states are recursively estimated at every time instance.

Since sensor measurements are in its local polar coordinate system, the nonlinear coordinate transformation between the local polar and global Cartesian coordinates introduces an intrinsic nonlinearity in the sensor registration problem. The approaches mentioned above adopted different kinds of approximation schemes to deal with the nonlinearity, i.e., using linear expansions to approximate some nonlinear functions \cite{leung1994least,Dana1990Registration,Cowley1993Registration,Fortunati2011Least,fischer1980registration,zhou1997an,helmick1993removal,Ristic2003Sensor,Fortunati2013On,Nabaa1999Solution,Okello2004Joint,Lin2006Multisensor,Li2010Joint,Taghavi2016practical} 
which usually leads to an uncontrolled model mismatch or using a large number of discrete samples (particles) to approximate {\black{some nonlinear functions}} which often requires an expensive computational cost when the number of sensors is large \cite{Zia2008EM,mahler2011bayesian,ristic2013calibration}. {\black{In addition to the schemes of dealing with}} the nonlinearity, {\black{many of the above works}} only consider the synchronous case, i.e., all sensors simultaneously measure the position of the target at the same time instance, which may not be always satisfied in many practical scenarios. 

Recent works in~\cite{pu2017two,Pu2018} proposed a semidefinite relaxation (SDR) based block coordinate descent (BCD) optimization approach to address the nonlinearity issue in the asynchronous sensor registration problem in the 2-dimensional scenario. The approach assumes the presence of one reference target moving at a nearly constant velocity (e.g., commercial planes and drones in the real world) and is effective in handling the nonlinearity while also guaranteeing an exact recovery of all sensor biases in the noiseless case. Consequently, it is interesting to explore such BCD algorithm in the 3-dimensional scenario, where more kinds of biases are involved and the estimation problem becomes even more highly nonlinear, i.e., more trigonometric functions with respect to different biases are multiplied together in the coordinate transformation. Furthermore, the BCD algorithm in~\cite{pu2017two,Pu2018} needs to solve convex semidefinite programing (SDP) subproblems at each iteration, whose computational cost is $\mathcal{O}(M^{4.5})$ (where $M$ is the total number of sensors). Such a high computational complexity limits the practical use of the BCD algorithm though it can well deal with the nonlinearity issue. An interesting question is, instead of solving SDPs, whether there is an algorithm which has a lower computational complexity and can solve the original nonconvex problem with a guaranteed convergence. This paper provides a positive answer to the above question.

\subsection{Our Contributions}
In this paper, we consider the 3-dimensional sensor registration problem, where more kinds of biases and noises are involved. Few works tackle the 3-dimensional registration problem except some notable works \cite{helmick1993removal,Ristic2003Sensor,Fortunati2011Least,Fortunati2013On,Jiang2018}. However, all of these works assume that sensors work synchronously. In this paper, we consider a practical scenario where all sensors are biased and work asynchronously. The major contributions are summarized as below:

\noindent
$\bullet$ \textbf{3-dimensional asynchronous sensor registration:} We consider the sensor registration problem in a general scenario, where sensors work asynchronously in 3-dimension and all sensors are biased.  Such a scenario is of great practical interest but has not been well studied yet. Compared with most of existing researches which consider the 2-dimensional scenario, there are more kinds of sensor biases in the 3-dimensional scenario, i.e.,  polar measurement biases and orientation angle biases, which increases the intrinsic nonlinearity in the bias estimation problem.
	
\noindent
$\bullet$ \textbf{A weighted nonlinear least squares formulation:} By exploiting the prior knowledge that the target moves with a nearly constant velocity, we propose a weighted nonlinear LS formulation for the 3-dimensional asynchronous multi-sensor registration problem, which takes different practical uncertainties into consideration, i.e., noise in sensor measurements and dynamic maneuverability in the target motion. The proposed formulation incorporates the covariance information of different practical uncertainties {\black{by appropriately choosing the weight matrix}}, which {\black{enables it to efficiently handle situations with different uncertainties.}}
	
\noindent
$\bullet$ \textbf{A computationally efficient BCD algorithm:} We exploit the special structure of the proposed weighted nonlinear LS formulation and develop a computationally efficient BCD algorithm for solving it. The proposed algorithm alternately updates different kinds of biases by solving linear LS problems and nonconvex quadratically constrained quadratic program (QCQP) problems. Compared with the previous work~\cite{Jiang2018} which utilizes the SDR-based technique for solving the nonconvex QCQP subproblems with a computational complexity of $\mathcal{O}(M^{4.5})$, we develop a low computational cost algorithm based on the alternating direction method of multipliers (ADMM) for solving those QCQP subproblems with a global optimality guarantee (under mild conditions), which only requires a computational complexity of $\mathcal{O}(M^{2})$.

\subsection{Organization and Notations}
The organization of this paper is as follows. In Section \ref{sec:problem_formulation}, we introduce basic models of sensor measurements and the target motion. In Section \ref{sec:ls}, we propose a nonlinear LS formulation for estimating all kinds of biases. To effectively solve the proposed nonlinear LS formulation, we develop a BCD algorithm in Section \ref{sec:bcd}. Numerical simulation results are presented in Section \ref{sec:simu}.

We adopt the following notations in this paper. Normal case letters, lower case letters in bold, and upper case letters in bold represent scalar numbers, vectors, and matrices, respectively. $\mathbf{X}^T$ and $\mathbf{X}^{-1}$ represent the transpose of matrix $\mathbf{X}$ and the inverse of invertible matrix $ \mathbf{X}$, respectively. $[\bm{X}]_{i:j,m:n}$ represents the matrix formed by elements of matrix $\bm{X}$ from rows $i$ to $j$ and columns $m$ to $n$. $\mathbf{x}_{n}$ denotes the $n$-th component of $\mathbf{x}$ and $\|\mathbf{x}\|$ denotes the Euclidean norm of $\mathbf{x}$. $\mathbb{E}_{\mathbf{w}}\{\cdot\}$  is the expectation operator with respect to random variable $\mathbf{w}$. $\mathbb{R}^{M}$  represents the set of $M$-dimensional real vectors. 
Other notions will be explained when they appear for the first time.

\section{Sensor Measurement Model}
\label{sec:problem_formulation}
Consider a 3-dimensional multi-sensor system with $M$ ($M>1$) sensors located at different known positions $\bm{p}_m\in\mathbb{R}^3,\forall\,m$. Assume there is a reference target moving with a nearly constant velocity in the space\footnote{Such a reference target can be selected from the branches of civilian airplanes.} and sensors measure the relative polar coordinates (i.e., range, azimuth, and elevation) at different time instances in an \emph{asynchronous} mode, i.e., polar coordinates of the target at different time instances are measured by different sensors. After a time period in which a set of local measurements are collected by sensors, all sensors' measurements are sent to a fusion center and compactly, all sensors' measurements are sorted and mapped onto a time axis, indexed by time instance $k=1,2,\ldots,K$. Without loss of generality,  we make the following (mild) regularity assumption.
\begin{assp}\label{assp:regular}
Assume that only one sensor observes the target at each time instance and every sensor has at least one measurement in the whole observation interval. In addition, the total number of observations is greater than the number of sensors, i.e., $K\geq M+1$. 
\end{assp}

 At time instance $k$, let the corresponding sensor be indexed by 
$s_k\in \{  1,2,\ldots,M \}$ and $\bm{\xi}_k=(x_k, y_k, z_k)$ be the position of the target in the \emph{global} Cartesian coordinate system. Then this target's position in sensor $s_k$'s local Cartesian coordinate system is denoted as $\bm{\xi}_k^{\prime}\in\mathbb{R}^3$, which can be expressed as a function with respect to $\bm{\xi}_k$ given as:
\begin{equation}\label{eq:pos_local_bias}
\begin{aligned}
\bm{\xi}_k^{\prime} = R^{-1}(\bm{\zeta}_{s_k}+\Delta\bm{\zeta}_{s_k})(\bm{\xi}_k - \bm{p}_{s_k}),
\end{aligned}
\end{equation}
where
\begin{equation}\label{eq: rotation}
\small
\begin{aligned}
\scriptsize
&R(\bm{\zeta})=R(\alpha,\beta,\gamma)= R_x(\alpha)R_y(\beta)R_z(\gamma) \\ =
&
\setlength{\arraycolsep}{1pt}
\left[
\begin{matrix}
1 & 0 & 0 \\
0 & \cos\alpha &-\sin\alpha \\
0 & \sin\alpha & \cos\alpha
\end{matrix}
\right] 
\left[
\begin{matrix}
\cos\beta & 0 & \sin\beta \\
0 & 1 & 0 \\
-\sin\beta & 0 & \cos\beta
\end{matrix}
\right]
\left[
\begin{matrix}
\cos\gamma & -\sin\gamma & 0 \\
\sin\gamma & \cos\gamma & 0 \\
0 & 0 & 1
\end{matrix}
\right].
\end{aligned}
\end{equation}
In the above,  $R(\alpha,\beta,\gamma)$ is a $3\times3$ \textit{rotation matrix}, which represents the coordinate rotation from the local to the global Cartesian coordinate systems,
and parameters $\alpha,\beta,\gamma\in[-\pi,\pi]$ there are rotation angles corresponding to \textit{roll, pitch}, and \textit{yaw}, respectively. Parameter $\bm{\zeta}_{m}=(\alpha_{m},\beta_{m},\gamma_{m})$ in \eqref{eq:pos_local_bias} is the \textit{presumed} rotation angles which is measured by sensor $m$'s orientation system and $\Delta \bm{\zeta}_{m}=(\Delta \alpha_{m},\Delta \beta_{m},\Delta \gamma_{m})$ is the unknown rotation angle biases. Note that $\Delta\bm{\zeta}_{m}$ is due to the imperfection of the orientation system of sensor $m$ and we call them the orientation biases in this paper. Hence, the term $\bm{\zeta}_{s_k}+\Delta\bm{\zeta}_{s_k}$  represents the \textit{true} rotation angles of sensor $s_k$.

In the following, we will discuss another category of sensor biases arising from imperfections in the sensor measurement system. Define the Cartesian-to-polar coordinate transformation function $h(x,y,z)$ as
\begin{equation*}
     h(x,y,z)= \begin{bmatrix} {\rho}\\ {\phi} \\ {\eta} \end{bmatrix}\triangleq\begin{bmatrix}  \sqrt{(x)^2 + (y)^2 + (z)^2} \\ \arctan(x, y)\\ \arctan( \sqrt{(x)^2 + (y)^2},z)\end{bmatrix},
\end{equation*}
where $\rho$, $\phi$, and $\eta$ are polar coordinates corresponding to \textit{range}, \textit{azimuth}, and \textit{elevation}, respectively. Then based on \eqref{eq:pos_local_bias} and the fact $R^{-1}(\cdot)=R^{T}(\cdot)$, the polar coordinates of $\bm{\xi}_k^\prime$ measured by sensor $s_k$ is denoted as $\bm{z}_k=({\rho}_k, {\phi}_k, {\eta}_k),$ which can be expressed as
\begin{equation}\label{eq:mea_model}
\begin{aligned}
\bm{z}_k &=h\left(\bm{\xi}_k^\prime\right) - \Delta\bm{z}_{s_k} + \bm{w}_k\\
&=h\left(R^T\left(\Delta\bm{\zeta}_{s_k}+\bm{\zeta}_{s_k}\right)(\bm{\xi}_k - \bm{p}_{s_k})\right) - \Delta\bm{z}_{s_k} + \bm{w}_k.
\end{aligned}
\end{equation}
In~\eqref{eq:mea_model}, $\Delta\bm{z}_{m}=(\Delta\rho_{m},\Delta\phi_{m},\Delta\eta_{m})$ is the measurement biases of sensor $m;$ $\bm{w}_k\in\mathbb{R}^3$ is the zero-mean Gaussian noise with
$$\bm{w}_k\sim\mathcal{N}(\bm{0},\textrm{diag}([\sigma_\rho^2,\sigma_\phi^2,\sigma_\eta^2])),$$
where $\sigma_\rho^2,\sigma_\phi^2$, and $\sigma_\eta^2$ are variances corresponding to range, azimuth, and elevation, respectively.

In view of the measurement model in~\eqref{eq:mea_model}, there are two types of biases, i.e., orientation biases $\Delta\bm{\zeta}_{s_k}$ and measurement biases $\Delta\bm{z}_{s_k}$. Both of them affect the expression for the global Cartesian coordinates $\bm{\xi}_k$. For each sensor, there are in total six biases, i.e., $\Delta \alpha_{m},\Delta \beta_{m},\Delta \gamma_{m},\Delta\rho_{m},\Delta\phi_{m}$, and $\Delta\eta_{m}$. Each of these biases has a different impact on $\bm{z}_k,$ and the geometric illustration of them is shown in Fig.~\ref{fig:coord}. We note that there is an intrinsic ambiguity among these six kinds of biases, and  this is formally stated in Proposition~\ref{prop:amg}.
\begin{figure}[h]
	\centering 
	\includegraphics[width=0.95\linewidth]{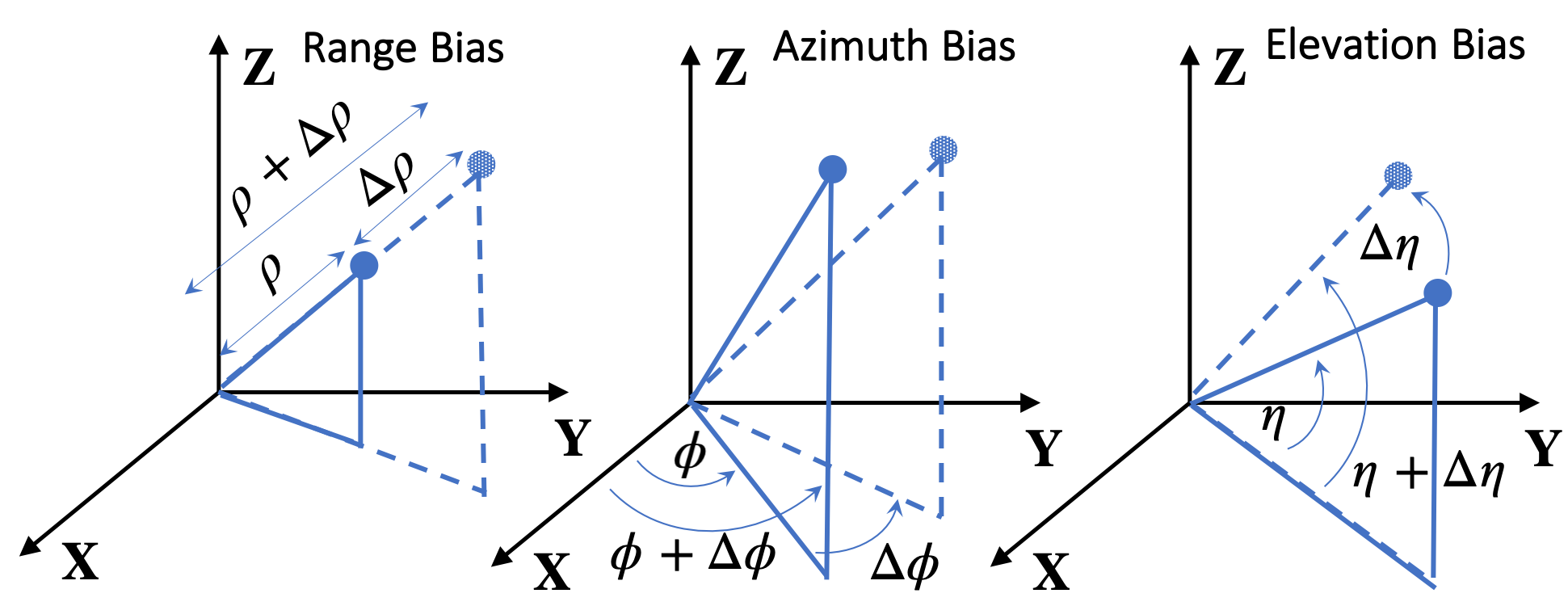}
	\includegraphics[width=0.95\linewidth]{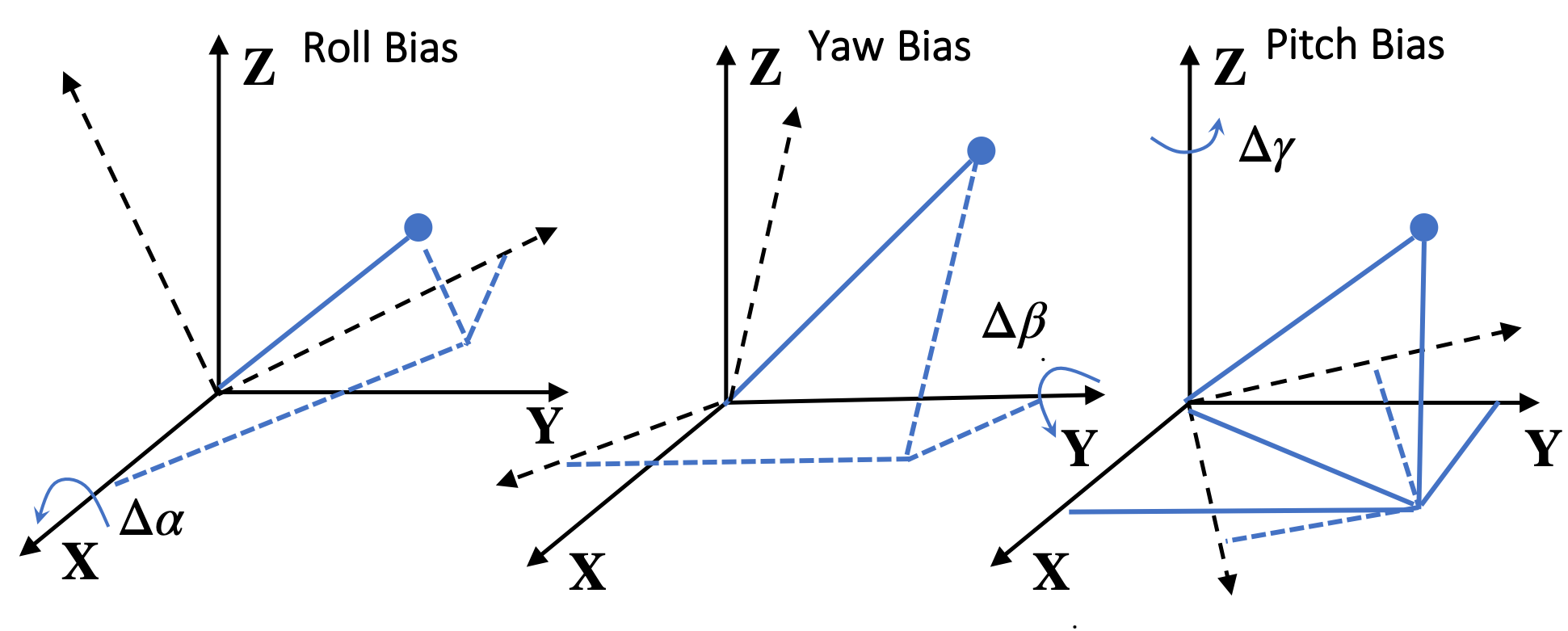}
	\caption{An illustration of different kinds of sensor biases.}
	\label{fig:coord}
\end{figure}
\begin{prop}\label{prop:amg}
	Let $c\in\mathbb{R}$, then any pair $(\Delta\gamma_{s_k},\Delta \phi_{s_k})$ satisfying $\Delta\gamma_{s_k}+\Delta \phi_{s_k}=c$ gives the same $\bm{z}_k$ in \eqref{eq:mea_model}.
\end{prop}
\begin{proof}
    See Appendix~\ref{app:propamg}.
\end{proof}
{\black 
Proposition~\ref{prop:amg} states that biases $\Delta \gamma$ and $\Delta\phi$ cannot be distinguished since they are intrinsically coupled in the measurement model \eqref{eq:mea_model}. Only their sum affects the measurements and they can be treated as one kind of bias \cite{Fortunati2011Least}. In the rest part of this paper, we fix $\Delta\phi_m=0,\forall\,m$ and only consider estimating $\Delta \gamma_m,\forall\,m$. 

Consequently, the goal of the {\black{3-dimensional}} multi-sensor registration problem is to determine sensor biases 
$$\bm{\theta}_{m}=[\Delta\rho_m,\Delta\eta_m,\Delta\alpha_m,\Delta\beta_m,\Delta\gamma_m],\forall\,m,$$from sensor measurements $\{\bm{z}_k\}^K_{k=1}$. It is worth noting that determining the sensor biases solely from \eqref{eq:mea_model} is a challenging task since the target position $\bm{\xi}_k$ at time instance $k$ is unknown. Any combination of $\bm{\theta}{s_k}$ and $\bm{\xi}_k$ can satisfy \eqref{eq:mea_model}. However, by leveraging the fact that the target moves with nearly constant velocity, the registration problem can be addressed.
}

\section{A Weighted Nonlinear LS Formulation}
\label{sec:ls}
In this section, we first introduce the unbiased spherical-to-Cartesian coordinate transformation to represent local noisy measurements $\{\bm{z}_k\}$ in the common Cartesian coordinate system as well as the nearly-constant-velocity motion model for the target. Then, we propose a weighted nonlinear least squares formulation for estimating sensor biases. 

According to the sensor measurement model in \eqref{eq:mea_model}, the target position $\bm{\xi}_k$ can be expressed by a nonlinear transformation with respect to $\bm{{z}}_k$ and $\bm{\theta}_{s_k}$, given in Proposition \ref{prop:unbias_mea}.

\begin{prop}[Unbiased Coordinate Transformation \cite{Song1998Unbiased}]\label{prop:unbias_mea}
	Given the measurement model in \eqref{eq:mea_model}, and suppose $\bm{\theta}_{s_k}$ take the value of the true bias. Then, we have 
	\begin{equation}\label{eq:unbias}
	R(\bm{\zeta}_{s_k}+\Delta \bm{\zeta}_{s_k}){\bar{h}^{-1}}(\bm{{z}}_k+\Delta \bm{{z}}_{s_k}) + \bm{p}_{s_k}+\bm{\varepsilon}_k=\bm{\xi}_k ,
	\end{equation}
	where $\bm{\varepsilon}_k$ is a zero-mean random noise and 
	$$
	{ \bar{h}^{-1}(\bm{z}_k)} = 
	\left[
	\begin{matrix}
	\lambda_\phi^{-1} \lambda_\eta^{-1}  \rho_k \cos\phi_k \cos\eta_k    \\
	\lambda_\phi^{-1} \lambda_\eta^{-1}  \rho_k \sin\phi_k \cos\eta_k \\
	\lambda_\eta^{-1}  \rho_k  \sin\eta_k 
	\end{matrix}
	\right]
	$$
	is the unbiased polar-to-Cartesian transformation function with $\lambda_\phi= e^{-\sigma_\phi^2/2} $ and $\lambda_\eta = e^{-\sigma_\eta^2/2} $ being the compensation factors. 
\end{prop}

Note that $\bm{\varepsilon}_k$ in Proposition~\ref{prop:unbias_mea} compensates the effects of the Gaussian noise $\bm{w}_k$ (in $\bm{z}_k$) after the nonlinear mapping $R(\bm{\zeta}_{s_k}+\Delta \bm{\zeta}_{s_k})\bar{h}^{-1}(\cdot)$ and $\bm{\varepsilon}_k$ can be viewed as a function of $\bm{w}_k$. The compensation factors ($\lambda_\phi$ and $\lambda_\eta$) in $\bar{h}^{-1}(\cdot)$ ensures  $\mathbb{E}_{\bm{w}_k}[ \bm{\varepsilon}_k]=\bm{0}$ but  $\bm{\varepsilon}_k$ is no longer a Gaussian noise due to the nonlinear mapping.

Next, assume that the reference target moves with a nearly constant velocity \cite{covariance}, i.e., 
\begin{equation}\label{eq:motion}
\begin{aligned}
\bm{\xi}_{k+1} &= \bm{\xi}_k + T_k \dot{\bm{\xi}}_k+ \bm{n}_k,\\
\dot{\bm{\xi}}_{k+1} &= \dot{\bm{\xi}}_k+ \dot{\bm{n}}_k,  
\end{aligned}
\end{equation}
where $\bm{\xi}_k=[x_k,y_k,z_k]^T\in\mathbb{R}^3$ and $\dot{\bm{\xi}}_k=[\dot{x}_k,\dot{y}_k,\dot{z}_k]^T\in\mathbb{R}^3$ are the position and velocity of the target at time instance $k$, respectively, $T_k\geq0$ is the interval time between time instances $k$ and $k+1$, and $\bm{n}_k\in\mathbb{R}^3$ and $\dot{\bm{n}}_k\in\mathbb{R}^3$ are the motion process noises for position and velocity at time instance $k$, respectively. The process noises $\bm{n}_k$ and $\dot{\bm{n}}_k$ obey the Gaussian distribution
\begin{equation}
\bm{n}_k\sim\mathcal{N}(\bm{0},qT_k^3/3\bm{I}_3),\ \dot{\bm{n}}_k\sim\mathcal{N}(\bm{0},qT_k\bm{I}_3).
\end{equation}
In the above, $q$ is the value of the noise power spectral density of the target motion and $\mathcal{N}(\bm{\mu},\bm{\Sigma})$ denotes the Gaussian distribution with mean $\bm{\mu}$ and covariance $\bm{\Sigma}$. 

Combining Proposition \ref{prop:unbias_mea} with the target motion model in \eqref{eq:motion}, {\black{we can immediately connect the}} measurements at time instances $k$ and $k+1$ with each other, given in Proposition \ref{prop:motion}.
{\black 
\begin{prop} \label{prop:motion}
	Given the measurement model in \eqref{eq:mea_model} and the target motion model in \eqref{eq:motion}, and consider time instances $k$ and $k+1$. Then, under Assumption~\ref{assp:regular}, we have 
	\begin{equation*}
	\begin{aligned}
	&g_{k}(  \boldsymbol{\theta}_{s_{k}})+\bm{\varepsilon}_{k}+T_{k}\dot{\bm{\xi}}_{k}+{\bm{n}}_k
	= g_{k+1}(  \boldsymbol{\theta}_{s_{k+1}})+\bm{\varepsilon}_{k+1},
	\end{aligned}
	\end{equation*}
	where $g_{k}( \bm{\theta}_{s_k}) = 
R(\bm{\zeta}_{s_{k}}+ { \Delta\bm{\zeta}_{s_{k}} }){ \bar{h}^{-1}}(\bm{z_}{k}+\Delta\bm{z}_{s_{k}}) + \bm{ p}_{s_k}$.
\end{prop}
}

By Proposition \ref{prop:motion}, we can establish the following nonlinear equations for the sensor biases estimation problem: 
\begin{equation}\label{eq:bias_mea_model}
\begin{aligned}
\bm{\epsilon}_k&=g_{k+1}(  \boldsymbol{\theta}_{s_{k+1}}) - 
g_{k}( \boldsymbol{\theta}_{s_k}) - T_k\dot{\bm{\xi}}_k,\ \forall\,k,\\
\dot{\bm{n}}_k &= \dot{\bm{\xi}}_{k+1} - \dot{\bm{\xi}}_k,\  \forall\,k,
\end{aligned}
\end{equation}
where $\bm{\epsilon}_k=\bm{\varepsilon}_{k+1}-\bm{\varepsilon}_{k}-\bm{n}_k$ is also a zero-mean random noise. 

Equations in~\eqref{eq:bias_mea_model} is a combination of sensor measurement model~\eqref{eq:mea_model} and target motion model~\eqref{eq:motion}. They can be regarded as the constructed measurement models for parameters $\{ \bm{\theta}_m\}$ and $\{ \dot{\bm{\xi}}_k\}$, where 
$\{\bm{\epsilon}_k\}$ and $\{ \dot{\bm{n}}_k \}$ are the zero-mean noise.

Based on \eqref{eq:bias_mea_model}, we propose to minimize the squares of $\bm{\epsilon}_k$ and $\dot{\bm{n}}_k$, which results  in the following nonlinear least squares formulation for estimating sensor biases:
\begin{equation}\label{eq:opt_nls}
\min \limits_{\{ \bm{\theta}_m \},\{  \dot{\bm{\xi}}_k \}}
\sum_{k=1}^{K-1}
\left\| \begin{bmatrix}g_{k+1}( \boldsymbol{\theta}_{s_{k+1}} )\\ \dot{\bm{\xi}}_{k+1} \end{bmatrix}
-\begin{bmatrix} g_{k}( \boldsymbol{\theta}_{s_k}) \\ \dot{\bm{\xi}}_{k} \end{bmatrix} - 
\begin{bmatrix} T_k\dot{\bm{\xi}}_k\\ \bm{0} \end{bmatrix} \right\|_{\bm{Q}_k}^2,
\end{equation}
where $\| \cdot \|_{\bm{Q}}$ denotes the weighted $\ell_2$ norm associated with the positive definite symmetric matrix $\bm{Q}$, i.e., $\| \bm{x}\|_{\bm{Q}}^2=\bm{x}^T\bm{Q}\bm{x}$. Notice that the LS form in~\eqref{eq:opt_nls} is derived from the nearly constant velocity model~\eqref{eq:motion}, where the term $ g_{k+1}(  \boldsymbol{\theta}_{s_{k+1}})-g_{k}(  \boldsymbol{\theta}_{s_{k}})-T_k\dot{\bm{\xi}}_{k}$ represents the mismatch of the position of the target in the global coordinate system and the other term $\dot{\bm{\xi}}_{k+1} - \dot{\bm{\xi}}_{k}$ corresponds to the mismatch of the nearly constant velocity. The current LS formulation cannot be directly applied to handle reference targets with strong maneuverability, such as those performing a constant turn motion. However, the fundamental concept of linking measurements at different time instances using the maneuver motion model is still relevant and can be further explored for more practical scenarios. 

The choices of the weight matrix $\bm{Q}_k$ in \eqref{eq:opt_nls} could be either the identity matrix or the (approximate) covariance matrix of $\bm{\epsilon}_k$ and $\dot{\bm{n}}_k$. More specifically, with $\bm{Q}_k=\bm{I},\ \forall\,k$, problem \eqref{eq:opt_nls} is a classical nonlinear LS problem based on the zero-mean property of noises $\bm{\epsilon}_k$ and $\dot{\bm{n}}_k$. Instead, let $\mathbf{R}_k$ be the \emph{approximated} covariance matrix with respect to $\bm{\epsilon}_k$ given in \cite{Song1998Unbiased}, then problem \eqref{eq:opt_nls} with 

\begin{equation}\label{eq:Qk}
\bm{Q}_k=\begin{bmatrix}
\bm{R}_{k+1}+\bm{R}_k+qT_k^3/3\bm{I} &qT_k^2/2\bm{I}\\
qT_k^2/2\bm{I} & qT_k\bm{I}
\end{bmatrix}^{-1}
\end{equation}
becomes a weighted NLS estimation. It is worth mentioning that two comprehensive works~\cite{duan2004comments,bordonaro2014decorrelated} have provided different formulations for evaluating $\bar{h}(\cdot)$ and $\bm{R}_k$ to resolve the incompatibility issue in~\cite{Song1998Unbiased}. However, as the focus of this paper is to utilize optimization techniques to handle the nonlinearity in the bias estimation problem, we have chosen to use the basic version of the unbiased coordinate transformation [24]. Nevertheless, the optimization techniques developed in this paper are also applicable to other unbiased coordinate transformations.


{\black{The mathematical difficulty of solving problem \eqref{eq:opt_nls} (no matter how to choose $\bm{Q}_k$) lies in its nonlinearity, as $g_k(\cdot)$ contains multiple products of trigonometric functions. In the next section, we will exploit the special structures of problem \eqref{eq:opt_nls} and develop a BCD algorithm for solving it (with any positive definite matrix $\bm{Q}_k$).}}

\section{A Block Coordinate Descent Algorithm}
\label{sec:bcd}
The proposed BCD algorithm iteratively minimizes the objective function with respect to one type of the biases with the others being fixed. The intuition of such decomposition is that different types of biases affect the measurements differently. By dividing the optimization variables into multiple blocks according to the type of bias, we can leverage the distinct mathematical effects of each type of biases on problem~\eqref{eq:opt_nls}. In particular, each subproblem with respect to one type of biases has a compact form, which can be solved uniquely under mild conditions. This enables us to develop a computationally efficient algorithm.

Let $\bm{v}=[\dot{\bm{\xi}}_1,\dot{\bm{\xi}}_2,\ldots,\dot{\bm{\xi}}_{K}]^T$, $\Delta\bm{\rho}=[\Delta\rho_1,\Delta\rho_2,\ldots,\Delta\rho_M]^T,$
and similar to the other kinds of biases. The objective function in \eqref{eq:opt_nls} can be expressed as $f(\bm{v},\Delta \bm{\rho}, \Delta \bm{\eta}, \Delta \bm{\alpha}, \Delta \bm{\beta}, \Delta \bm{\gamma})$, and the proposed BCD algorithm at iteration $t$ is 
\begin{subequations}\label{eq:bcd}
	\begin{align}
	&\bm{v}^{t+1} = \arg \min \limits_{\bm{ v}} f( \bm{v},\Delta \bm{\rho}^t, \ldots, \Delta \bm{\gamma}^{t}), \label{eq:sub_xi} \\
	&{\Delta \bm{\rho}}^{t+1} = \arg \min \limits_{\Delta\bm{ \rho}} f(\bm{v}^{t+1},\Delta \bm{\rho},  \ldots, \Delta \bm{\gamma}^{t}),\label{eq:sub_rho} \\
	&{\Delta \bm{\eta}}^{t+1} = \arg \min \limits_{\Delta\bm{ \eta}} f(\bm{v}^{t+1},\ldots, \Delta \bm{\eta}, \ldots, \Delta \bm{\gamma}^{t}), \label{eq:sub_eta} \\
	&{\Delta \bm{\alpha}}^{t+1} = \arg \min \limits_{\Delta\bm{ \alpha}} f(\bm{v}^{t+1},\ldots,\Delta \bm{\alpha}, \ldots, \Delta \bm{\gamma}^{t}),\label{eq:sub_alp} \\
	&{\Delta \bm{\beta}}^{t+1} = \arg \min \limits_{\Delta\bm{ \beta}} f(\bm{v}^{t+1},\ldots, \Delta \bm{\beta}, \Delta \bm{\gamma}^{t} ),\label{eq:sub_beta}\\ 
	&{\Delta \bm{\gamma}}^{t+1} = \arg \min \limits_{\Delta\bm{ \gamma}} f(\bm{v}^{t+1}, \ldots,\Delta \bm{\beta}^{t+1}, \Delta \bm{\gamma} ).\label{eq:sub_gamma} 
	\end{align}
\end{subequations}
In the next, we derive the solutions to the above subproblems \eqref{eq:sub_xi}--\eqref{eq:sub_gamma}. For simplicity, we will omit the iteration index $t$ in the coming derivation. 

\subsection{LS Solution for Subproblems \eqref{eq:sub_xi} and \eqref{eq:sub_rho}}
By \eqref{eq:opt_nls}, we  have the following convex quadratic reformulation for subproblem \eqref{eq:sub_xi}:
\begin{equation}\label{eq:opt_v}
 \min \limits_{\bm{v}}\ \| \bm{H}_v\bm{v} + \bm{c}_{v}\|_{\bm{Q}}^2, 
\end{equation}
where $\bm{H}_{v} \in \mathbb{R}^{6(K-1)\times3K}$ is the coefficient matrix, $\bm{ c}_{v} \in \mathbb{R}^{6(K-1)}$ is the constant vector, and $\bm{Q}\in \mathbb{S}_+^{6(K-1)}$ is a positive definite matrix. Detailed expressions of them are given in Appendix \ref{app:deriv}-A of the Supplementary Material. {Setting the gradient to be zero, the} closed-form solution to subproblem \eqref{eq:sub_xi} is
\begin{equation}\label{eq:opt_solu_v}
\bm{ v}^{*}
= -(\bm{H}_v^T \bm{Q} \bm{H}_v)^{{\dagger}} \bm{H}_v^T\bm{Q}_{{v}}\bm{ c}_v,
\end{equation}
where $\dagger$ denotes the pseudo-inverse operator.

The subproblem \eqref{eq:sub_rho} for updating $\Delta\bm{\rho}$ also has a closed-form solution since its objective is also convex quadratic with respect to $\Delta\bm{\rho}$. Subproblem \eqref{eq:sub_rho} can be reformulated as
\begin{equation}\label{eq:opt_rho}
\min \limits_{\Delta\bm{\rho}}\  \| \bm{H}_{\rho} \Delta\bm{\rho}
+ \bm{ c}_{{\rho}}\|_{\bm{Q}}^2,
\end{equation}
where $\bm{H}_{\rho} \in \mathbb{R}^{6(K-1)\times M}$ is the coefficient matrix, and $\bm{ c}_{\rho} \in \mathbb{R}^{6(K-1)}$ is the constant vector. Detailed expressions of them are given in  Appendix \ref{app:deriv}-B of the Supplementary Material. Similar to subproblem \eqref{eq:sub_xi},  the closed-form solution of \eqref{eq:sub_rho} is
\begin{equation}\label{eq:opt_solu_rho}
\Delta \bm{\rho}^{*}
= -(\bm{H}_\rho^T \bm{Q}\bm{H}_\rho)^{{ {\dagger}}} \bm{H}_\rho^T\bm{Q}\bm{ c}_\rho. 
\end{equation}

\subsection{Solution for Subproblems \eqref{eq:sub_eta}--\eqref{eq:sub_gamma}}
Since $g_k(\cdot)$ in \eqref{eq:opt_nls} is linear with respect to the trigonometric functions of the angle biases, all the four subproblems \eqref{eq:sub_eta}-\eqref{eq:sub_gamma} can be reformulated into a same QCQP form. 
For better presentation, we use $\Delta \bm{\vartheta}\in \Omega=\{ \Delta\bm{\eta},\Delta\bm{\alpha},\Delta\bm{\beta},\Delta\bm{\gamma} \}$ to denote one kind of these angle biases and the corresponding subproblem can be reformulated as 
\begin{equation}\label{eq:qcqp}
\begin{aligned}
\min_{ \bm{x} \in \mathbb{R}^{2M} }\  \|\bm{H}_{\vartheta} \bm{x} + \bm{ c}_{\vartheta}\|_{\bm{Q}}^2,\quad
\textrm{s.t.}\ \bm{x} \in \mathcal{C},
\end{aligned} 
\end{equation}
where 
\begin{equation*}
\bm{x} = [\cos\Delta\vartheta_1,\sin\Delta \vartheta_1,\ldots,\cos\Delta \vartheta_M ,\sin\Delta \vartheta_M]^T
\end{equation*}
and set $\mathcal{C}$ is defined as 
$$\mathcal{C}\triangleq\{  \bm{x}\in\mathbb{R}^{2M}\mid x_{2i}^2+x_{2i-1}^2 = 1,~i=1,\dots,M\}.$$
In \eqref{eq:qcqp}, $\bm{H}_{\vartheta}\in \mathbb{R}^{6(K-1) \times 2M}$ is the coefficient matrix and $\bm{ c}_{\vartheta} \in \mathbb{R}^{6(K-1)}$ is the constant vector. Detailed expressions of them are given in Appendix \ref{app:deriv}-C of the Supplementary Material.

Problem \eqref{eq:qcqp} is a non-convex problem due to its non-convex equality constraints. In general, such class of problem is known to be NP-hard \cite{luo20104} and  the SDR based approach \cite{luo2010semidefinite,lu2019tightness} was used to solve it. In particular, in \cite{Pu2018}, the authors utilized the SDR technique to solve a specific case of problem \eqref{eq:qcqp} where it can be reformulated as a complex QCQP, and  proved that the SDR is tight, i.e., the global minimizer of the QCQP problem can be obtained by solving the SDR under a mild condition. This motivates us to ask whether problem~\eqref{eq:qcqp} can also be globally solved and whether there are some efficient algorithms better than solving the SDR. 

Note that the non-convex constraints in problem~\eqref{eq:qcqp} enjoy an ``easy-projection'' property. In particular, let $\textrm{Proj}_\mathcal{C}(\bm{x})$ denote the projection of any non-zero point $\bm{x}\in\mathbb{R}^{2M}$ onto $\mathcal{C}$. Then $\textrm{Proj}_\mathcal{C}(\bm{x})$ admits the following closed-form solution: 
\begin{equation*}
\textrm{Proj}_\mathcal{C}(\bm{x})=\arg\min_{ \bm{z} \in \mathcal{C}}\| \bm{x}-\bm{z}\|^2=\bm{D}_{\bm{x}}^{-1}\bm{x},
\end{equation*}
where $\bm{D}_{\bm{x}}$ is a diagonal matrix defined as 
$$\bm{D}_{\bm{x}}=\textrm{Diag}\left(\left[\sqrt{\bm{x}_1^2+\bm{x}_2^2},\ldots,\sqrt{\bm{x}_{2M-1}^2+\bm{x}_{2M}^2}\right]\right)\otimes\bm{I}_2,$$
and $\otimes$ denotes the Kronecker product. Such ``easy-projection'' property together with recent results of the ADMM for nonconvex problems~\cite{wang2019global} provides us a potential way of efficiently solving problem \eqref{eq:qcqp} with a guaranteed convergence. The choice of the ADMM algorithm over other potential methods such as the gradient projection (GP) method is discussed in Remark~\ref{rmk:admm}. In the next, we will apply the ADMM for solving it and also show that the ADMM is able to globally solve problem \eqref{eq:qcqp} under mild conditions.  

\begin{rmk}[ADMM versus GP]\label{rmk:admm}
Note that the GP method is a good option for solving problem~\eqref{eq:qcqp}, whose convergence to a stationary point can be guaranteed under the Kurdyka-Lojasiewicz framework~\cite{attouch2013convergence}. In this paper, we choose the ADMM rather than the GP method for solving problem~\eqref{eq:qcqp} is because the ADMM shows a more stable numerical behavior than the GP method and in particular the ADMM is more stable to the choice of the step size than the GP method. This is particularly important for solving our interested problem \eqref{eq:opt_nls} as we need to solve subproblems in the form of~\eqref{eq:qcqp} many times in the proposed BCD algorithm. 
\end{rmk}

Based on such ``easy-projection'' property, we introduce an auxiliary variable $\bm{x}=\bm{z}$ to split problem \eqref{eq:qcqp}, which makes all subproblems in the ADMM iterations {\black{admit closed-form solutions}}. In particular, subproblem \eqref{eq:qcqp} is reformulated as
\begin{equation}\label{eq:qcqp_ref}
\begin{aligned}
\min_{ \bm{z},\bm{x} } \ \|\bm{H}_{\vartheta} \bm{x} + \bm{ c}_{\vartheta}\|^2+\iota_{\mathcal{C}}(\bm{z}),\quad
\textrm{s.t.}\ \bm{x}=\bm{z},
\end{aligned} 
\end{equation}
where $\iota_{\mathcal{C}}(\bm{z})$ is the indicator function for set $\mathcal{C}$, given as 
\begin{equation}\label{eq:indfunc}
\iota_{\mathcal{C}}(\bm{z})=\left\lbrace
\begin{aligned}
0,\ &\textrm{if }\bm{z}\in\mathcal{C},\\
\infty, \ &\textrm{otherwise}.
\end{aligned}
\right.
\end{equation}
The augmented Lagrangian function for problem \eqref{eq:qcqp_ref} is 
\begin{equation}\label{eq:Lrho}
L_\rho(\bm{x},\bm{z},\bm{\lambda})=\| \bm{H}_{\vartheta}\bm{x}+\bm{c}_{\vartheta}\|^2+\bm{\lambda}^T(\bm{x}-\bm{z})+\frac{\rho}{2}\|\bm{x}-\bm{z} \|^2+\iota_{\mathcal{C}}(\bm{z}),
\end{equation}
where $\rho>0$ is the penalty parameter and $\bm{\lambda}\in\mathbb{R}^{2M}$ is the Lagrange multiplier for the equality constraint $\bm{x}=\bm{z}$. Then, the ADMM updates at  iteration $\ell$ are
\begin{subequations}\label{eq:admm}
\begin{align}
\bm{x}^{\ell+1}&=\min_{ \bm{x}}L_\rho(\bm{x},\bm{z}^\ell,\bm{\lambda}^\ell),\label{eq:admm_subx}\\
\bm{z}^{\ell+1}&=\min_{ \bm{z}}L_\rho(\bm{x}^{\ell+1},\bm{z},\bm{\lambda}^\ell),\label{eq:admm_subz}\\
\bm{\lambda}^{\ell+1}&=\bm{\lambda}^t+\rho(\bm{x}^{\ell+1}-\bm{z}^{\ell+1}).
\end{align}
\end{subequations}
Note that subproblem \eqref{eq:admm_subx} is an unconstrained strongly convex quadratic problem and subproblem \eqref{eq:admm_subz} is actually a projection problem, both of which can be solved in closed forms. Details are given in Algorithm \ref{alg:admm}.

\begin{algorithm}[h]
	\caption{ADMM for Solving Problem \eqref{eq:qcqp_ref}}  
	\begin{algorithmic}[1]  
		\STATE Initialize $\bm{x}^0,\bm{z}^0,\bm{\lambda}^0$ and $\rho>0$
		\FOR{$\ell=0,1,2,\ldots,$}
		\STATE $\bm{x}^{\ell+1}=-(\bm{H}_{\vartheta}^T\bm{Q}\bm{H}_{\vartheta}+\frac{\rho}{2}\bm{I})^{-1}(2\bm{H}_{\vartheta}^T\bm{Q}\bm{c}_{\vartheta}+\bm{\lambda}^\ell-\rho\bm{z}^\ell)/2$;
		\STATE $\bm{z}^{\ell+1}=\textrm{Proj}_{\mathcal{C}}(\frac{\bm{\lambda}^\ell+\rho\bm{x}^{\ell+1}}{\rho})$;
		\STATE $\bm{\lambda}^{\ell+1}=\bm{\lambda}^\ell+\rho(\bm{x}^{\ell+1}-\bm{z}^{\ell+1})$;
		\ENDFOR
		\STATE Output $\bm{x}^*,\bm{z}^*$
	\end{algorithmic}  
	\label{alg:admm}
\end{algorithm} 

Next, we discuss the convergence property of Algorithm \ref{alg:admm}. Before doing this, we first study the structures of coefficients $\bm{H}_\vartheta$ and $\bm{c}_\vartheta$ in problem \eqref{eq:qcqp}. Note that $\bm{H}_\vartheta$ and $\bm{c}_\vartheta$ can be viewed as random matrices/vectors and the randomness is caused by those zero-mean Gaussian noises, i.e., $\bm{w}_k$ in~\eqref{eq:mea_model} and $\bm{n}_k,\,\dot{\bm{n}}_k$ in~\eqref{eq:motion}. To understand the effects of these noises, let us suppose that there is no noise, i.e., $\bm{w}_k=\bm{0},\bm{n}_k=\dot{\bm{n}}_k=\bm{0}$, for all $ k$, and all the other biases $\Omega\setminus\Delta\bm{\vartheta}$ and velocity $\bm{v}$ take the true values. Then by the same derivations in Appendix \ref{app:deriv} of the Supplementary Material (with compensation factors $\lambda_\eta=\lambda_\phi=1$), we attain new $\bar{\bm{H}}_\vartheta$ and $\bar{\bm{c}}_\vartheta$. Intuitively, $\bar{\bm{H}}_\vartheta$ and $\bar{\bm{c}}_\vartheta$ represent the ``true'' parts of $\bm{H}_\vartheta$ and $\bm{c}_\vartheta$ by removing measurement noise $\bm{w}_k$, motion process noises $\bm{n}_k$ and $\dot{\bm{n}}_k$, and errors of other biases and velocity. Consequentially, the ``errors'' in $\bm{H}_\vartheta$ and $\bm{c}_\vartheta$ can be defined as 
\begin{equation}\label{eq:decomp}
\begin{aligned}
\Delta\bm{H}_\vartheta=\bm{H}_\vartheta-\bar{\bm{H}}_\vartheta,\ 
\Delta\bm{c}_\vartheta=\bm{c}_\vartheta-\bar{\bm{c}}_\vartheta.
\end{aligned}
\end{equation}
Technically, decomposing $\bm{H}_\vartheta$ and $\bm{c}_\vartheta$ in~\eqref{eq:decomp} provides a natural way of characterizing the convergence behavior of Algorithm \ref{alg:admm}, given in Theorem \ref{thm:admm}.

\begin{thm}\label{thm:admm}
	Consider applying Algorithm~\ref{alg:admm} to solve problem~\eqref{eq:qcqp_ref} with a sufficiently large penalty parameter $\rho$, then
	\begin{enumerate}
	    \item the sequence generated by Algorithm \ref{alg:admm} has at least one limit point and each limit point is a stationary point of $L_\rho(\bm{x},\bm{z},\bm{\lambda})$ (defined in~\eqref{appeq:sta});
	    \item if $\Delta\bm{H}_\vartheta$ and $\Delta \bm{c}_\vartheta$ in \eqref{eq:decomp} are sufficiently small, then $L_\rho(\bm{x},\bm{z},\bm{\lambda})$ has a unique stationary point $(\bar{\bm{x}},\bar{\bm{z}},\bar{\bm{\lambda}})$ and the entire sequence generated by Algorithm \ref{alg:admm} converges to  $(\bar{\bm{x}},\bar{\bm{z}},\bar{\bm{\lambda}});$ 
	    furthermore, if $\Delta\bm{H}_\vartheta=\bm{0}$ and $\Delta \bm{c}_\vartheta=\bm{0}$, then $\bar{\bm{x}}$ is equal to the corresponding true bias.
	\end{enumerate}

\end{thm}
\begin{proof}
	See Appendix \ref{app:thm}.
\end{proof}

The first statement of Theorem~\ref{thm:admm} is a direct result of \cite[Corollary 2]{wang2019global}, which ensures the convergence of the ADMM for a class of the nonconvex problem (with a Lipschitiz continuous gradient) over a compact manifold. 
The second statement of Theorem ~\ref{thm:admm} on the uniqueness of the stationary point is due to the special structure of problem~\eqref{eq:qcqp_ref}. The \emph{entire} sequence convergence in Theorem~\ref{thm:admm} is a combined result of the general convergence of the ADMM in \cite[Corollary 2]{wang2019global} and the special structure of problem \eqref{eq:qcqp_ref}. 

\subsection{Proposed BCD Algorithm}
{\black{
The proposed BCD algorithm for solving problem \eqref{eq:opt_nls} is summarized as Algorithm \ref{alg:BCD}. Three remarks on its convergence, computational complexity, and comparison with other BCD variants are given in order. 

\begin{algorithm}[h]  
	\caption{Proposed BCD Algorithm for Problem \eqref{eq:opt_nls}}  
	\begin{algorithmic}[1]  
		\REQUIRE Sensor measurements $\{\bm{z}_k\}^K_{k=1}$ 
		\STATE Initialize $\Delta\bm{\rho}^0,\Delta\bm{\eta}^0,\Delta\bm{\alpha}^0,\Delta\bm{\beta}^0,\Delta\bm{\gamma}^0$, and $\bm{v}^0$
		\FOR{$t=0,1,2,\ldots,$}
		\STATE Update $\bm{v}^{t+1}$ by \eqref{eq:opt_solu_v};
		\STATE Update $\Delta\bm{\rho}^{t+1}$ by \eqref{eq:opt_solu_rho};
		\STATE Sequentially update $\Delta\bm{\eta}^{t+1}$,$\Delta\bm{\alpha}^{t+1}$,$\Delta\bm{\beta}^{t+1}$,$\Delta\bm{\gamma}^{t+1}$ by Algorithm \ref{alg:admm}, respectively;
		\ENDFOR   
		\ENSURE Sensor biases $\Delta\bm{\rho}^*,\Delta\bm{\eta}^*,\Delta\bm{\alpha}^*,\Delta\bm{\beta}^*,\Delta\bm{\gamma}^*$.
	\end{algorithmic}  
	\label{alg:BCD}  
\end{algorithm}

\textbf{Convergence of Algorithm \ref{alg:BCD}.}
The well-known convergence analysis of BCD~\cite{bertsekas1999nonlinear,wright2015coordinate} shows that if all subproblems in the proposed BCD Algorithm~\ref{alg:BCD} have unique solutions and can be globally solved, then the iterates can converge to a stationary point of the non-convex problem~\eqref{eq:opt_nls}. The uniqueness and the global optimality of solutions $\bm{v}^*$ (in~\eqref{eq:opt_solu_v}) for subproblem~\eqref{eq:sub_xi} and $\Delta\bm{\rho}^*$ (in~\eqref{eq:opt_solu_rho}) for subproblem~\eqref{eq:sub_rho} can be guaranteed if $\bm{H}_v$ and $\bm{H}_{\rho}$ are of full rank. For subproblems~\eqref{eq:sub_eta}--\eqref{eq:sub_gamma}, the full rankness of $\bm{H}_{\vartheta}$ together with sufficiently small $\Delta\bm{H}_\vartheta$ and $\Delta \bm{c}_\vartheta$ in~\eqref{eq:decomp} further ensure that each of the corresponding subproblems has a unique solution and can be numerically attained by the proposed ADMM Algorithm~\ref{alg:admm} (cf. Theorem~\ref{thm:admm}). According to the definitions of $\bm{H}_v$, $\bm{H}_{\rho}$, and $\bm{H}_{\vartheta}$ in Appendix~\ref{app:deriv} of the  Supplementary Material and Assumption~\ref{assp:regular}, we know that (i) $\bm{H}_v$ in~\eqref{eq:opt_v} is always of full rank since each sensor has at least one measurement; (ii) $\bm{H}_\rho$ in~\eqref{eq:opt_rho} and $\bm{H}_\vartheta$ in~\eqref{eq:qcqp} are of full rank with probability one (whose proof is provided in Appendix D of the Supplementary Material). Based on the above discussion, we can conclude that Algorithm~\ref{alg:BCD} can converge to a stationary point of the non-convex problem~\eqref{eq:opt_nls} if those $\Delta\bm{H}_\vartheta$ and $\Delta \bm{c}_\vartheta$ along the BCD iterates are all sufficiently small. Finally, we remark that this sufficient small condition provides a theoretical understanding for the convergence of Algorithm~\ref{alg:BCD} and in general it depends on the initial point. Our numerical simulation results in Section~\ref{sec:simu} suggest that such a condition always holds in various scenarios by simply initializing all biases to be zeros. The results in the noiseless setting (Section~\ref{subsec:simulglobal}) even verify the global optimality of Algorithm~\ref{alg:BCD}. 


\textbf{Computational complexity of Algorithm \ref{alg:BCD}.} All subproblems in the proposed BCD updates \eqref{eq:sub_xi}--\eqref{eq:sub_gamma} take two special forms, i.e., unconstrained linear LS problems (i.e., subproblems \eqref{eq:opt_v} and \eqref{eq:opt_rho}) and quadratic programs with special non-convex constraints (i.e., subproblem \eqref{eq:qcqp}). It can be shown that the computational complexity of solving subproblem \eqref{eq:opt_v}, \eqref{eq:opt_rho}, and \eqref{eq:qcqp} is $\mathcal{O}(K^2+K)$, $\mathcal{O}(M^{2}+MK+M)$, and $\mathcal{O}\left(I\left((2M)^2+2MK+2M\right)\right)$ (where $I$ is the total iteration number of the ADMM), respectively. It is worthwhile mentioning that the computational complexity of solving \eqref{eq:qcqp} by using the ADMM, which is $\mathcal{O}(I(2M)^2),$  is generally much less than that of using the SDR technique, which is $\mathcal{O}((2M)^{4.5})$~\cite{luo2010semidefinite}. The computational time comparison  
in Table~\ref{tab:time} (Section~\ref{sec:simul:sub_time}) demonstrates the much better computational efficiency of using the ADMM over the SDR based approach for solving subproblem \eqref{eq:qcqp}.

\textbf{BCD variants.}
Finally, we remark that there are several nice extensions of the BCD framework such as the block successive minimization method~\cite{razaviyayn2013unified} and the maximum block improvement (MBI) method~\cite{aubry2018new}. Both of them can be applied to solve problem~\eqref{eq:opt_nls} with guaranteed convergence (with a suitable approximation of the objective function). The sequential optimization version of MBI~\cite{aubry2018new} can effectively handle nonconvex constraints. In this paper, we still use the basic version of BCD because the nonconvex BCD subproblems with respect to angle biases enjoy the unique solution property (cf. Theorem~\ref{thm:admm}), which reveals an interesting insight into the considered bias estimation problem that the true bias can be recovered under the ideal case (i.e., $\Delta\bm{H}_\vartheta=\bm{0}$ and $\Delta \bm{c}_\vartheta=\bm{0}$). MBI also enjoys the same property of the uniqueness of the subproblems but it requires more computational cost than classical BCD. A parallel implementation of MBI can improve its computational efficiency but this is beyond the scope of this paper. Simulation results in Section~\ref{subsec:simulglobal} further show that all biases are exactly recovered by our proposed BCD Algorithm~\ref{alg:BCD} in the noiseless case.

}}


\section{Simulation Results}
\label{sec:simu}
{\black{
In this section, we present simulation results to illustrate the performance of using our proposed BCD Algorithm~\ref{alg:BCD} to solve the proposed formulation~\eqref{eq:opt_nls}. 

\textbf{{ Scenario setting:}} We consider a scenario with four sensors and one target in the 3-dimensional space, as illustrated in Fig. \ref{fig:sc}. The locations and biases of each sensor are listed in Table \ref{tab:senbias}. The target moves with a nearly constant velocity $\mathbb{E}[ \dot{\bm{\xi}}_k ] = [0,0.3,0]\,$km/s with the initial position being $[-30,-5,8]\,$km. The presumed rotation angle $\bm{\zeta}_m$ is fixed at $(0^\circ,0^\circ,0^\circ)$ for all $m=1,2,3,4$. The sensors work in an asynchronous mode that four sensors measure the target positions every $10\,$s with different starting times, i.e., $2.5\,$s, $5\,$s, $7.5\,$s, $10\,$s, respectively. Moreover, each sensor sends $20$ stamped measurements to a fusion center for sensor registration and the total observation last $210\,$s. All simulations are done on a laptop (Intel Core i7) with Matlab2018b. 

\textbf{{ Baseline approaches:}} We select three representative approaches in the literature and extend them to the considered 3-dimensional case where all sensors are biased and asynchronous. The first approach is the augmented state Kalman filtering (ASKF) approach \cite{Zhou2004A} which treats sensor biases as the augmented states and uses the Kalman filter to jointly estimate sensor biases and target states. The second approach is  the LS approach \cite{Fortunati2011Least} which solves the proposed NLS formulation in \eqref{eq:opt_nls} using the { linear approximation} (referred as linearized LS). The third approach is the ML approach \cite{Ristic2003Sensor,Fortunati2013On}, which is originally proposed for the synchronous multi-sensor system \cite{Ristic2003Sensor} or for the scenario with one sensor being bias-free \cite{Fortunati2013On}. To extend it to the considered 3-dimensional case with a stable numerical behavior, we first estimate the target state by the smoothed Kalman filter (SKF) and then estimate sensor biases by the Gaussian-Newton method \cite{bertsekas1999nonlinear}. For simplicity, we refer this approach as SKF-GN. 

{ \textbf{Proposed algorithms' parameter selection:} In all of our following simulations, the parameters in the BCD and ADMM algorithms are selected as follows: the proposed BCD algorithm is initialized with all biases being zeros and the proposed ADMMs for solving subproblems \eqref{eq:sub_eta}--\eqref{eq:sub_gamma} are also initialized with the corresponding angle biases being zeros; the penalty parameter $\rho$ in the ADMM is set to be the average of the diagonal elements of matrix $\bm{H}^T_{\vartheta}\bm{Q}\bm{H}_{\vartheta}$ in \eqref{eq:qcqp}; the proposed BCD algorithm is terminated when the maximum difference between two consecutive BCD iterates is less than $10^{-5}$  and the ADMM is terminated when both primal and dual residuals are less than $10^{-9}$.
}\begin{table}
	\centering
	\caption{Sensor Positions and Biases [km, degree].}
	\begin{tabular}{c c c c c c c c}
		\toprule
		& Position    & $\Delta \rho$ &$\Delta \eta$ &$\Delta \alpha$& $\Delta \beta$&$\Delta \gamma$ \\ 
		\midrule
		Sensor 1 & $[0, -15, 0]$                  & -0.5             & -2           &   -2    & 1     & -1 \\ 
		Sensor 2 &  $[-20, 5, 2]$                 & 0.3             &    -2        &   2    &   -1 &   -1   \\
		Sensor 3 & $[20, 5, 0]$                     & -0.4         &  -2          &   2    &  -2   &2 \\ 
		Sensor 4 & $[0, 10, -1]$                       &   -0.2            &  -1          &   -2    &  -1 & 1    \\ 
		\midrule
	\end{tabular}
	\label{tab:senbias}
\end{table}
\begin{figure}
	\centering 
	\includegraphics[width=0.9\linewidth]{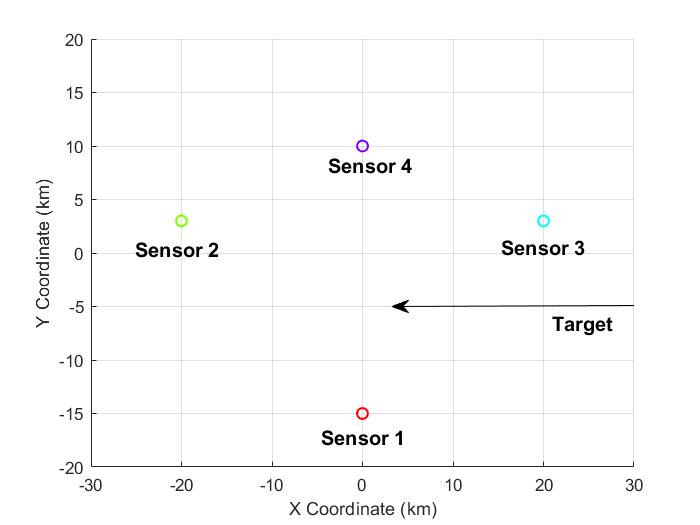}
	\caption{An illustration of the simulation scenario.}
	\label{fig:sc}
\end{figure}

{ 
\subsection{Global Optimality of Proposed BCD in the Noiseless Case}\label{subsec:simulglobal}
In this subsection, we test the convergence behavior of the proposed BCD for solving problem ~\eqref{eq:opt_nls} with $\bm{Q}_k=\bm{I}$, $\forall k$. We consider the noiseless case with the measurement noise variance $\sigma_\rho^2=0\,\textrm{km}$, $\sigma_\eta^2=\sigma_\phi^2=0^\circ$ and the target motion noise $q=0\,\textrm{m}^2/\textrm{s}^3$. In such a case, the true biases achieve a zero loss and hence they are a global solution of problem~\eqref{eq:opt_nls}. Fig.~\ref{fig:globalcovg} plots the optimization residuals over the number of iterations and the elapsed time. The optimization residuals for different kinds of biases are defined as the sum of the squared estimation error over all sensors, e.g., the optimization residual for the range bias is $\sum_{m=1}^M (\Delta\hat{\rho}_m - \Delta\rho_m)^2$, where $\Delta\hat{\rho}_m$ is the estimated range bias of sensor $m$. The convergence curves in Fig.~\ref{fig:globalcovg} demonstrate the global optimality of the proposed BCD algorithm in the noiseless case, i.e., the loss function converges to zero and the finally returned solution by the BCD algorithm is very close to the true biases. The similar convergence behavior has also been observed in many other scenarios. In particular, Fig.~\ref{fig:globalcovgall} shows the convergence curves of the proposed BCD algorithm in $20$ randomly generated scenarios\footnote{{ The sensors' positions are uniformly placed over the region $[-20,20]^2\times [-1,1]\,$km and sensors' biases are generated from the Gaussian distribution, i.e., $\Delta\rho\sim\mathcal{N}(0,(1\, \textrm{km})^2)$, and $\Delta\eta,\Delta\alpha,\Delta\beta,\Delta\gamma\sim\mathcal{N}(0,(3^\circ)^2)$.}}. 

\begin{figure}
	\centering 
	\includegraphics[width=1\linewidth]{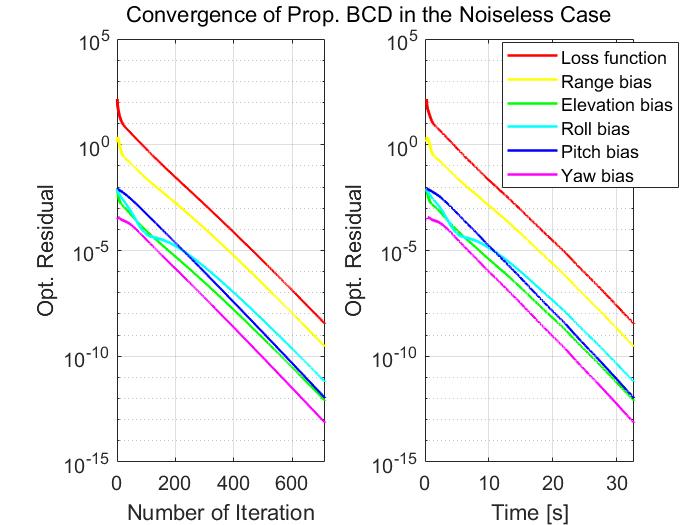}
	\caption{{Convergence of the proposed BCD algorithm in the noiseless case of the scenario illustrated in Fig. \ref{fig:sc}.}}
	\label{fig:globalcovg}
\end{figure}

\begin{figure}
	\centering 
	\includegraphics[width=1\linewidth]{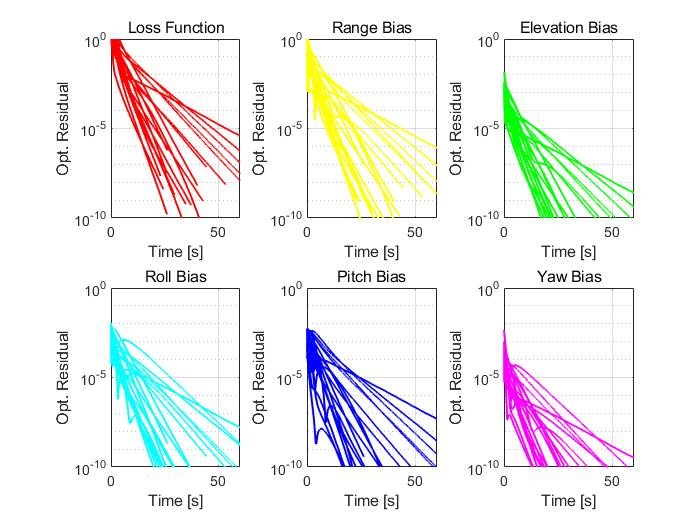}
	\caption{{ Convergence of the proposed BCD algorithm in $20$ randomly generated scenarios, where each curve corresponds to one randomly generated scenario.}}
	\label{fig:globalcovgall}
\end{figure}
}

\subsection{Estimation Performance in the Noisy Case}
{ In this subsection, we evaluate the estimation performance of the proposed BCD algorithm in the noisy case. 
Note that the global optimality in the noisy case is generally hard to examine, as the global solution of the corresponding problem is unknown. However, our numerical results show that the solution found by the BCD algorithm achieves a pretty satisfactory estimation performance.

In the rest of simulations, we use the root mean square error (RMSE) as the estimation performance metric. The root of hybrid Cramer-Rao lower bound (RHCRLB) \cite{Noam2009tightness,Fortunati2011Least} is used as the performance benchmark.} The RMSEs and RHCRLBs in all of the following figures are obtained by averaging 100 independent Monte Carlo runs. We will compare the proposed BCD algorithm with the aforementioned three approaches under different setups.
}}


\subsubsection{Estimation Performance Comparison}\label{sec:simul:sub_est}
{\black{
In this part, we compare the estimation performance of different approaches with $\sigma_\rho=0.05\,\textrm{km}$, $\sigma_\eta=\sigma_\phi=0.02^\circ$, and $q=0.5\,\textrm{m}^2/\textrm{s}^3$. { Two different choices of $\bm{Q}_k$ in the proposed formulation~\eqref{eq:opt_nls} are considered}: one is $\bm{Q}_k=\bm{I}$ which corresponds to classical nonlinear LS estimation (referred as BCD-NLS), and the other is $\bm{Q}_k$ in \eqref{eq:Qk} which can be regarded as pseudo ML estimation (referred as BCD-PML). The RMSEs of the range bias of different sensors are shown in Fig. \ref{fig:siml_bars} and the RMSEs of different kinds of angle biases are shown in Fig. \ref{fig:siml_bars2}. We can observe that the two proposed approaches (i.e., BCD-NLS and BCD-PML) achieve much smaller RMSEs than the other three approaches. The main reason for this is that all the other three approaches utilize the { linear approximation} to deal with the nonlinearity, and such approximation usually results into a model mismatch which degrades the estimation performance. Notice that the only difference between Linearized LS and BCD-NLS is the algorithm used for solving problem \eqref{eq:opt_nls} with $\bm{Q}=\bm{I}$. More specifically, Linearized LS utilizes the { linear approximation} which admits a closed-form solution for estimating biases while BCD-NLS uses the proposed BCD algorithm to deal with the nonlinearity. Clearly, BCD-NLS achieves smaller RMSE. Furthermore, by incorporating the second-order statistics, BCD-PML achieves a smaller RMSE than BCD-NLS, which demonstrates the effectiveness of the proposed choice of  $\bm{Q}_k$ in  \eqref{eq:Qk}.
}}
\begin{figure}
	\centering 
	\includegraphics[width=0.8\linewidth]{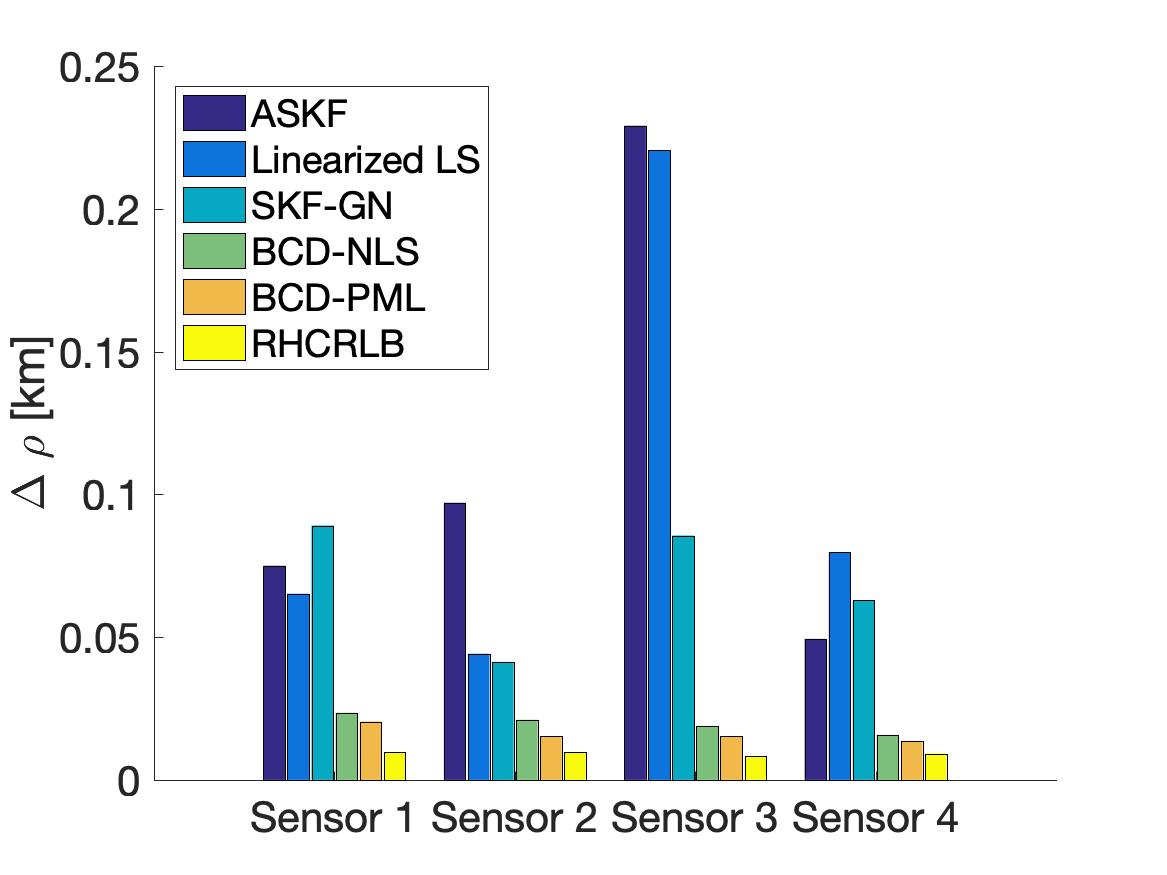}
	\caption{RMSE and RHCRLB for { $\sigma_\rho=0.05\,\textrm{km}$}, $\sigma_\eta=\sigma_\phi=0.02^\circ$, and $q=0.5\,\textrm{m}^2/\textrm{s}^3$.}
	\label{fig:siml_bars}
\end{figure}
\begin{figure*}
	\centering 
	\subfigure[Elevation bias.]{
	\includegraphics[width=0.23\linewidth]{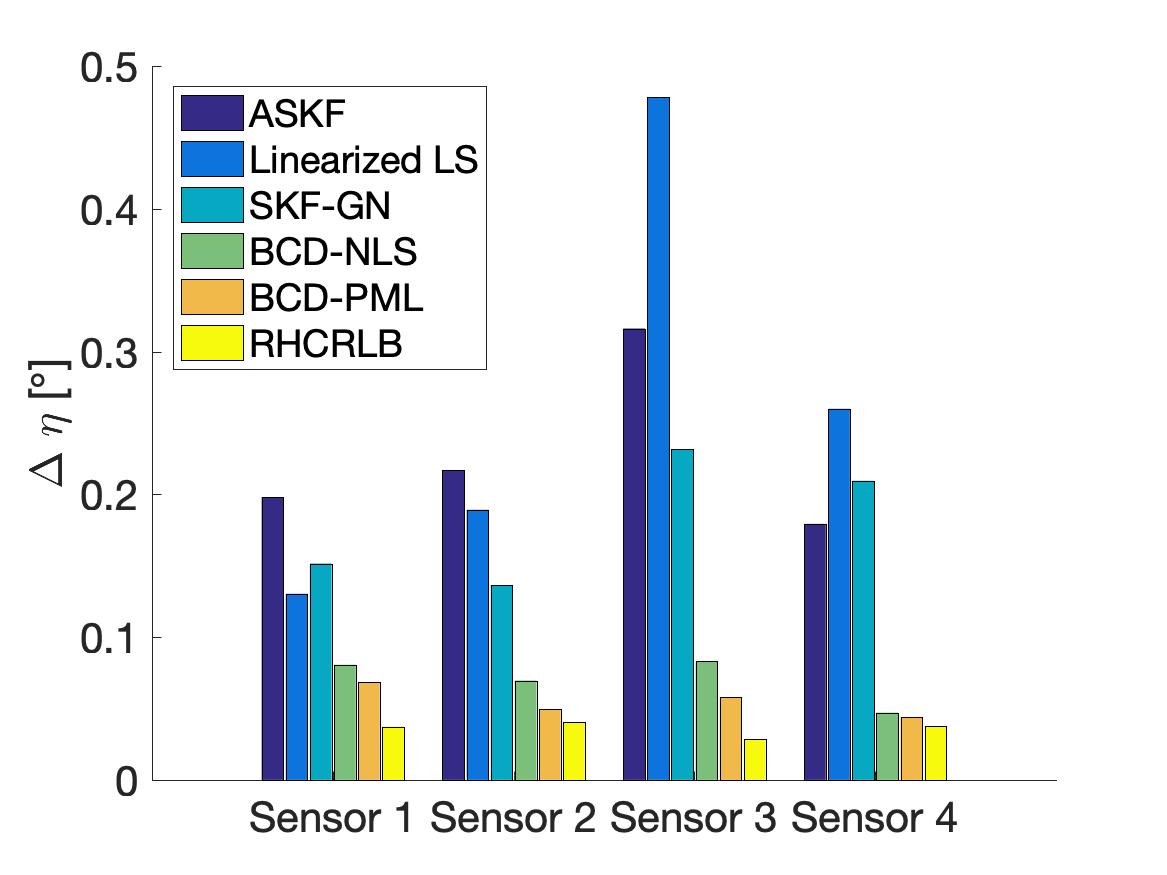}}
	\subfigure[Roll bias.]{
	\includegraphics[width=0.23\linewidth]{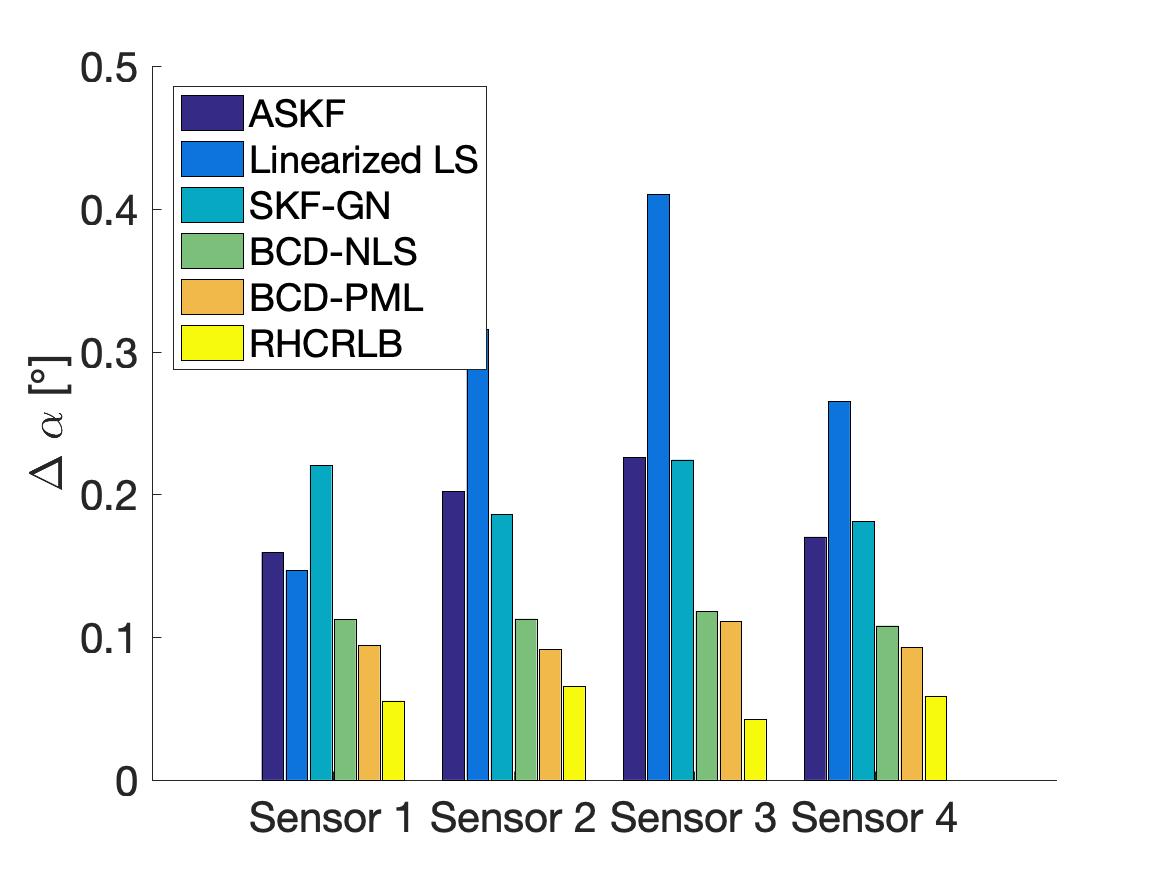}}
	\subfigure[Pitch bias.]{
	\includegraphics[width=0.23\linewidth]{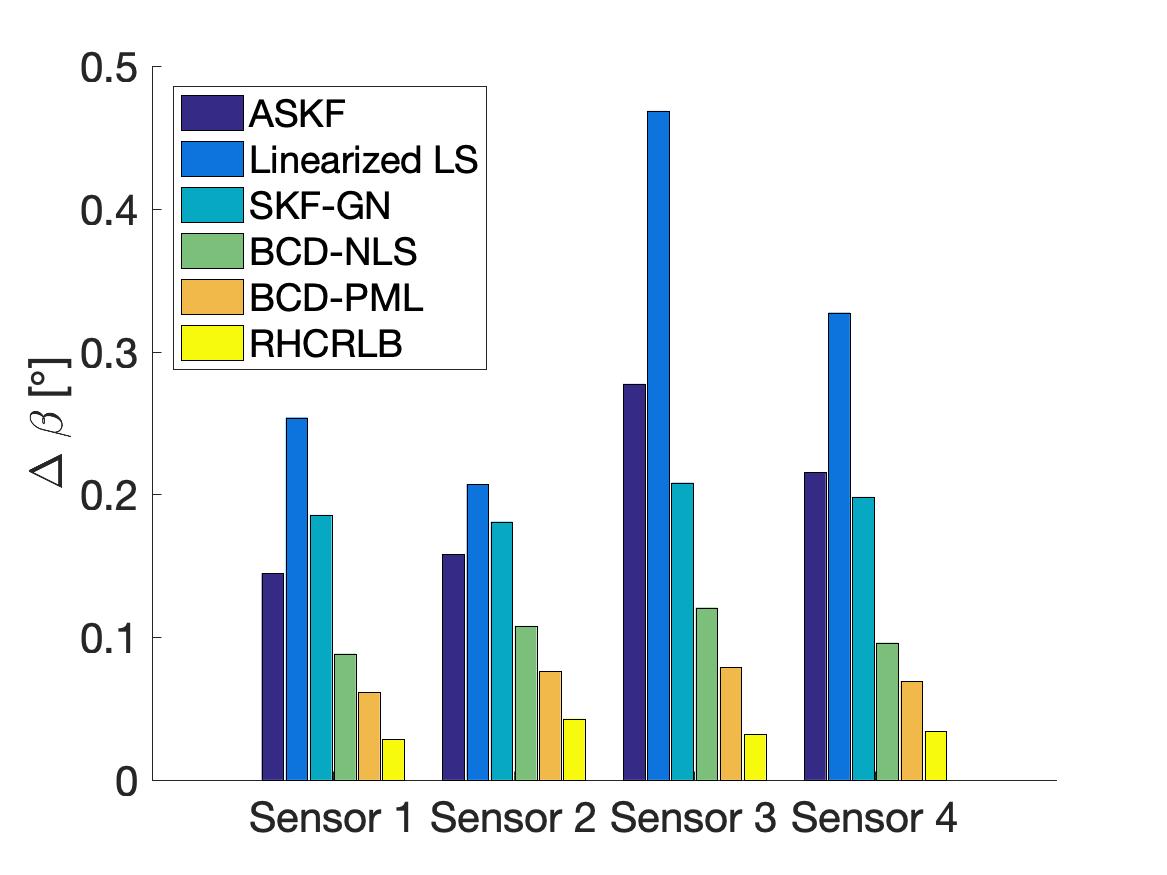}}
	\subfigure[Yaw bias.]{
	\includegraphics[width=0.23\linewidth]{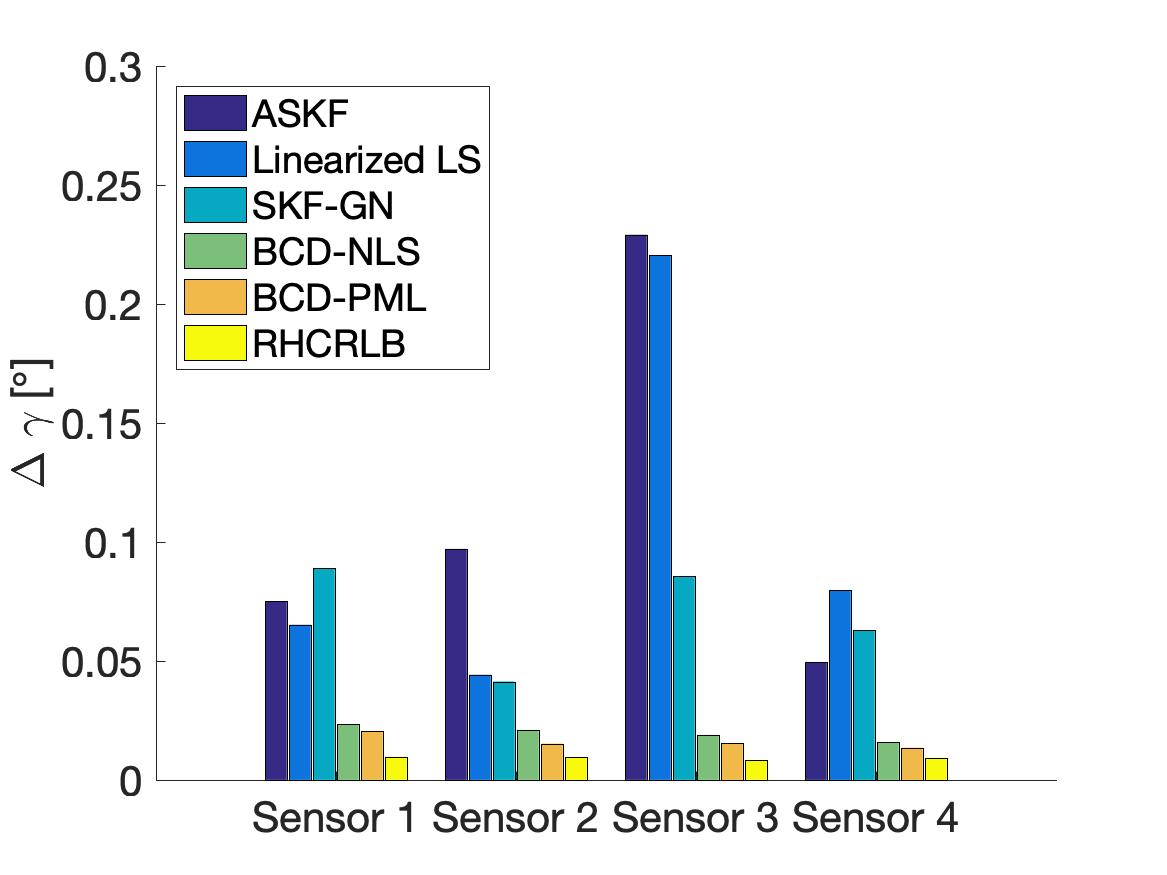}}
	\caption{RMSE and RHCRLB for $\sigma_\rho=0.05\,\textrm{km}$, $\sigma_\eta=\sigma_\phi=0.02^\circ$, and $q=0.5\,\textrm{m}^2/\textrm{s}^3$.}
	\label{fig:siml_bars2}
\end{figure*}

\subsubsection{Computational Time Comparison}\label{sec:simul:sub_time}
{\black{
In this part, we compare the computational time of different approaches. The averaged computational time of different approaches is summarized in Table~\ref{tab:time}. Notice that the SDR technique in \cite{Pu2018} can be extended to deal with the nonlinearity in the 3-dimensional scenario\footnote{{\black{Under mild conditions, the SDR technique gives the same solution of problem \eqref{eq:qcqp} obtained by the proposed ADMM. In particular, the proof in Appendix \ref{app:thm} can be modified to show that the SDR admits a unique rank-one solution if $\Delta\bm{H}$ and $\Delta\bm{c}$ in \eqref{eq:qcqp} are sufficiently small. This implies that the SDR technique and the ADMM give the same solution.}}}. We also include the computational time of employing the SDR technique to solve QCQP subproblems in form of \eqref{eq:qcqp} to demonstrate the computational effectiveness of the proposed ADMM. { Values in the bracket in Table \ref{tab:time} correspond to the computational time of the BCD algorithm where CVX \cite{cvx} is used to solve the SDRs of the corresponding QCQP subproblems.} In our simulations, we observe that the difference of the final solutions returned by the BCD algorithm with the SDR technique and the ADMM (for solving the subproblems) is less than or equal to $10^{-5}$. { This implies the two approaches return almost the same solutions}.

From Table~\ref{tab:time}, it can be observed that the { linear approximation} simplifies the estimation procedure which makes the three approaches, i.e., ASKF, Linearized LS, and SKF-GN, take much less computational time than the proposed two approaches. This is because all of these three approaches are based on the { linear approximation}, which simplify the estimation procedure but also degrades the estimation performance. Comparing our proposed algorithms with the ADMM and the SDR technique (in bracket) being used for solving the subproblems, we can observe that the proposed ADMM significantly reduces the computational time. We remark that sensor biases in practice usually change slowly \cite{Fortunati2011Least}, i.e., in hours or days,  and thus the computational time of the proposed approaches is acceptable for practical applications. 
}}
\begin{table}
	\centering
	\caption{Computational Time [second]}
	\begin{tabular}{c c c c c c}
		\toprule
		 ASKF    & Linearized LS & SKF-GN  & BCD-NLS&   BCD-PML\\ 
		\midrule
		0.03    & 0.02         &    0.16             &     34.61 (1309.21)     &  33.28 (1289.54) \\ 
		\midrule
	\end{tabular}
	\label{tab:time}
\end{table}
\begin{figure}
	\centering 
	\includegraphics[width=0.9\linewidth]{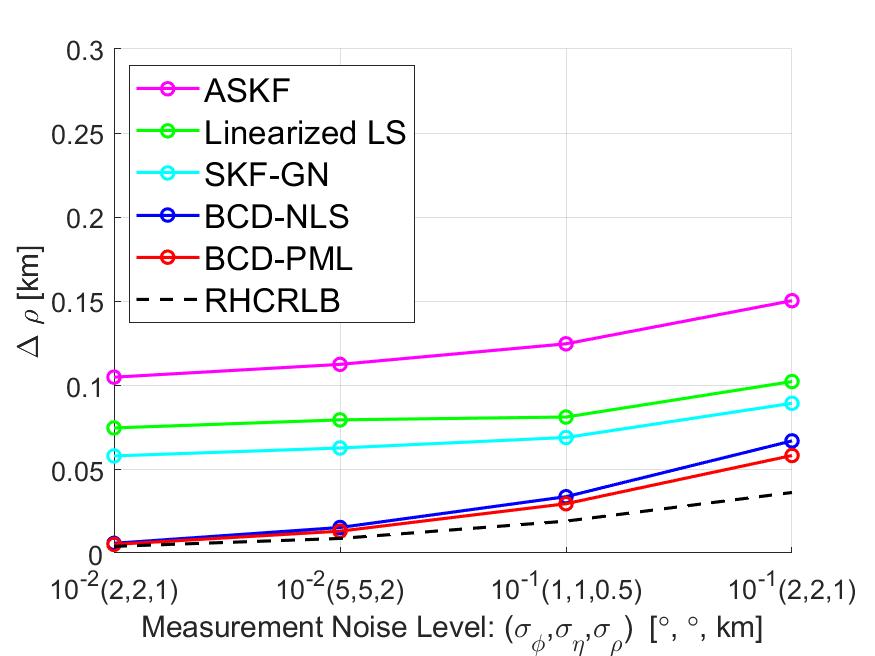}
	\caption{RMSE and RHCRLB of range and elevation biases with different measurement noise levels $(\sigma_\rho,\sigma_\phi,\sigma_\eta)$.}
	\label{fig:simul_mea}
\end{figure}
\subsubsection{Impact of Measurement Noise}\label{sec:simul:sub:mea}
In this part, we present some simulation results to illustrate the effect of the measurement noise on the estimation performance. We fix the motion noise density $q=0.5\,\textrm{m}^2/\textrm{s}^3$ and set different values of the measurement noise level $(\sigma_\rho,\sigma_\phi,\sigma_\eta)$. Since the measurement noise level has similar impacts on the estimation performance of the four sensors, we only present the averaged RMSEs and RHCRLBs over four sensors. The comparison of the range biases is presented in Fig. \ref{fig:simul_mea} and the comparison of other biases is presented in Fig. \ref{fig:simul_mea2}. It can be observed that BCD-PML achieves the smallest RMSEs among all cases and BCD-NLS is a little worse than BCD-PML. For the other three approaches, due to the model mismatch error introduced by the { linear approximation} procedure, they all suffer a ``threshold'' effect when the measurement noise is small, i.e., $\sigma_\phi=\sigma_\eta=0.02^\circ,\sigma_\rho=0.01\,\textrm{km}$. In contrast, the proposed BCD algorithm exploits the special problem structure and utilizes the nonlinear optimization techniques to deal with the nonlinearity issue. Therefore, BCD-NLS and BCD-PML do not have the `threshold' effects and their RMSEs are quite close to RHCRLBs when the measurement noise level is small.

\subsubsection{Impact of Target Motion Noise}\label{sec:simul:sub:mo}
In this part, we present some simulation results to illustrate the effect of the target motion noise on the estimation performance. We fix the measurement noise level to be $\sigma_\rho=0.1\,\textrm{km}$ and $\sigma_\phi=\sigma_\eta=0.1^\circ$ and change the values of the motion noise density $q$ in interval $[0.2,2]\,\textrm{m}^2/\textrm{s}^3$. Similar to Section \ref{sec:simul:sub:mea}, we compare the averaged RMSEs of the range in Fig. \ref{fig:simul_mo} and the RMSEs of other angle biases in Fig. \ref{fig:simul_mo2}. It can be observed that all approaches are somehow robust to the motion process noise. Similarly, we can observe that the proposed approaches perform better than the other three approaches and their RMSEs are very close to RHCRLBs.
\begin{figure}
	\centering 
	\includegraphics[width=0.9\linewidth]{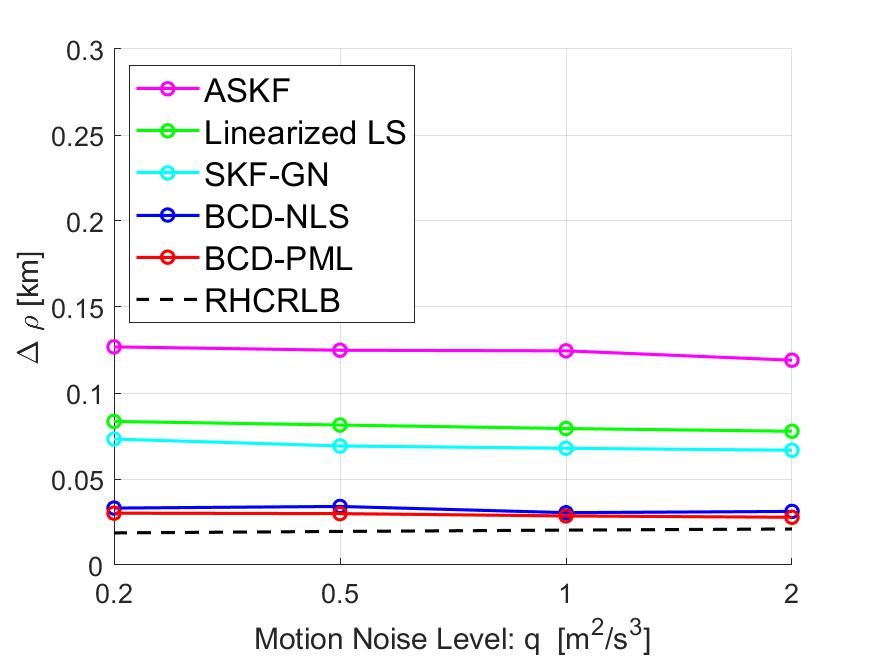}
	\caption{RMSE and RHCRLB range and elevation biases with different motion process noise levels $q$.}
	\label{fig:simul_mo}
\end{figure}

\subsubsection{Impact of Sensor Bias}\label{sec:simul:sub:bias}
To illustrate the effect of the bias level on the estimation performance (especially for the other three approaches which use the { linear approximation}), we present some simulation results with different levels of the sensor biases. Specifically, we fix { $\sigma_\rho=0.1\,\textrm{km}$}, $\sigma_\phi=\sigma_\eta=0.1^\circ$, and $q=0.5\,\textrm{m}^2/\textrm{s}^3$, and change the bias values in Table \ref{tab:senbias} by multiplying them with a positive constant $c$, i.e., $c=0.2,0.5,1,2$. The averaged RMSEs of the range biases and other angle biases are compared in Figs. \ref{fig:simul_bias} and \ref{fig:simul_bias2}, respectively. {\black{In the case where the bias level is small, i.e., $c=0.2$, all approaches have similar RMSEs since the { linear approximation} leads to a smaller model mismatch error. However, as $c$ increases, the RMSEs of the other three approaches increase very fast but the proposed two approaches are robust to the bias level.}}
\begin{figure}
	\centering 
	\includegraphics[width=0.9\linewidth]{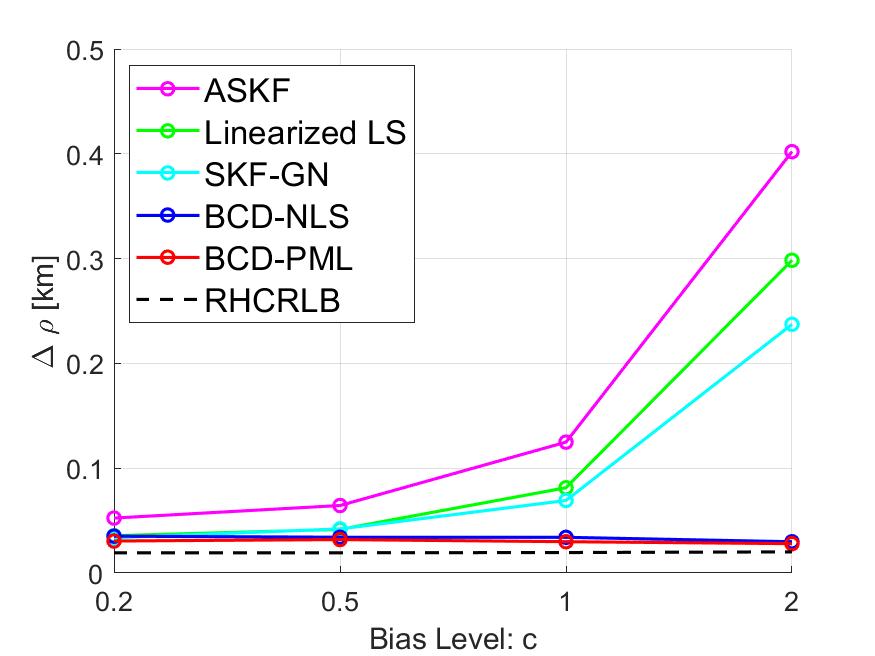}
	\caption{RMSE and RHCRLB range and elevation biases with different bias levels $c$.}
	\label{fig:simul_bias}
\end{figure}

\begin{figure*}
	\centering 
	\subfigure[Elevation bias.]{
		\includegraphics[width=0.23\linewidth]{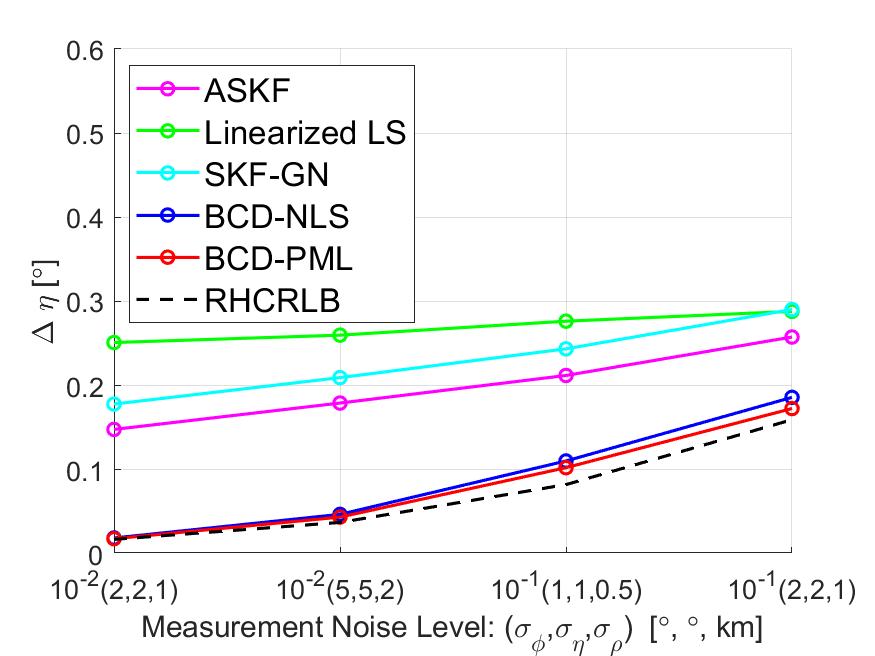}}
	\subfigure[Roll bias.]{
		\includegraphics[width=0.23\linewidth]{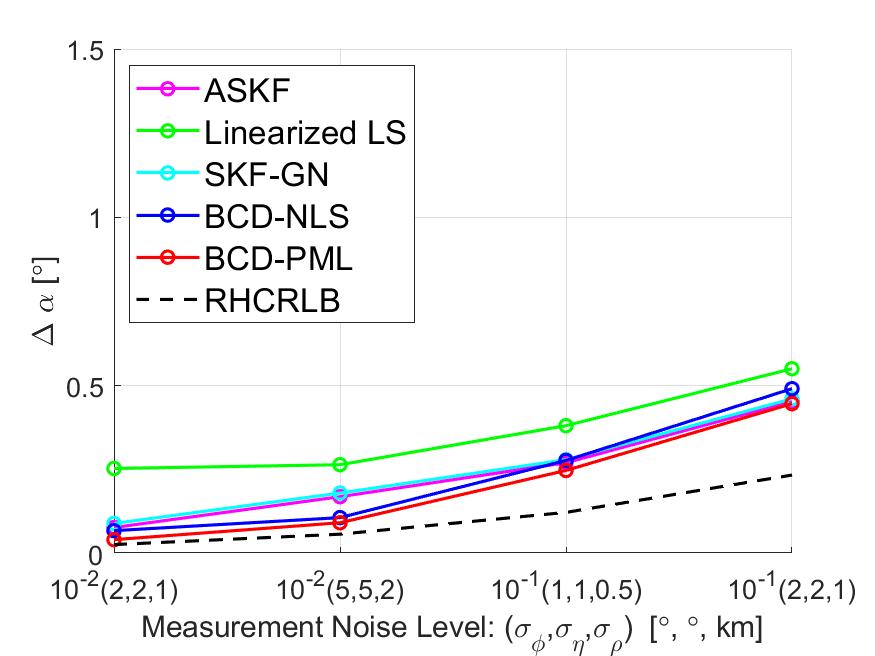}}
	\subfigure[Pitch bias.]{
		\includegraphics[width=0.23\linewidth]{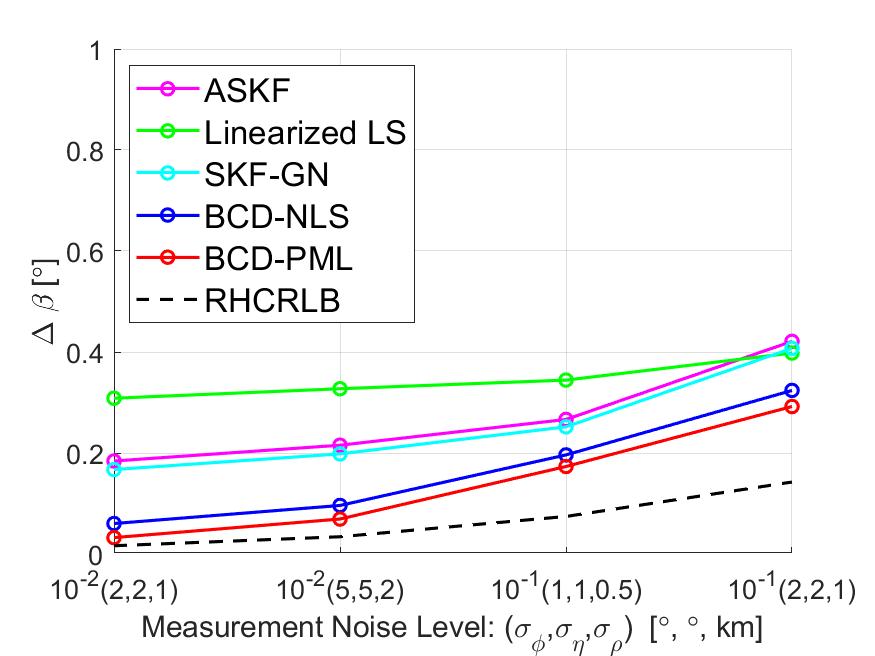}}
	\subfigure[Yaw bias.]{
		\includegraphics[width=0.23\linewidth]{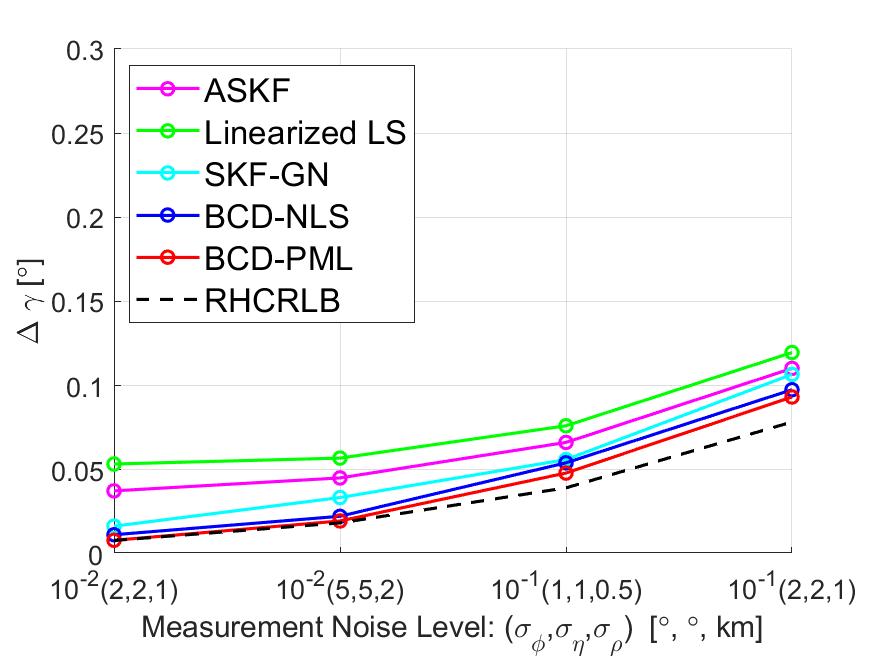}}
	\caption{RMSE and RHCRLB of rotation angle biases with different measurement noise levels $(\sigma_\rho,\sigma_\phi,\sigma_\eta)$.}
	\label{fig:simul_mea2}
\end{figure*}

\begin{figure*}
	\centering 
	\subfigure[Elevation bias.]{
		\includegraphics[width=0.23\linewidth]{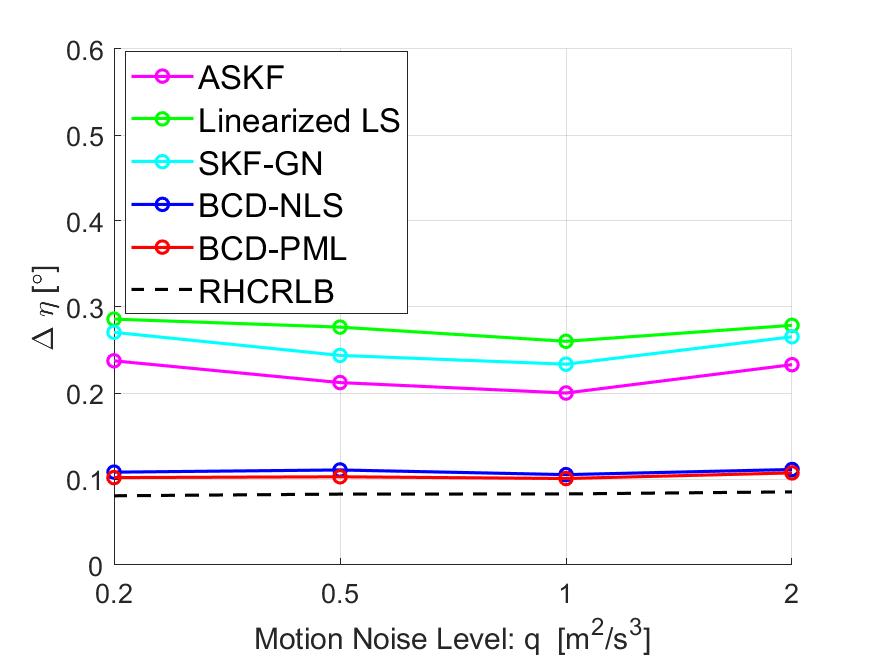}}
	\subfigure[Roll bias.]{
		\includegraphics[width=0.23\linewidth]{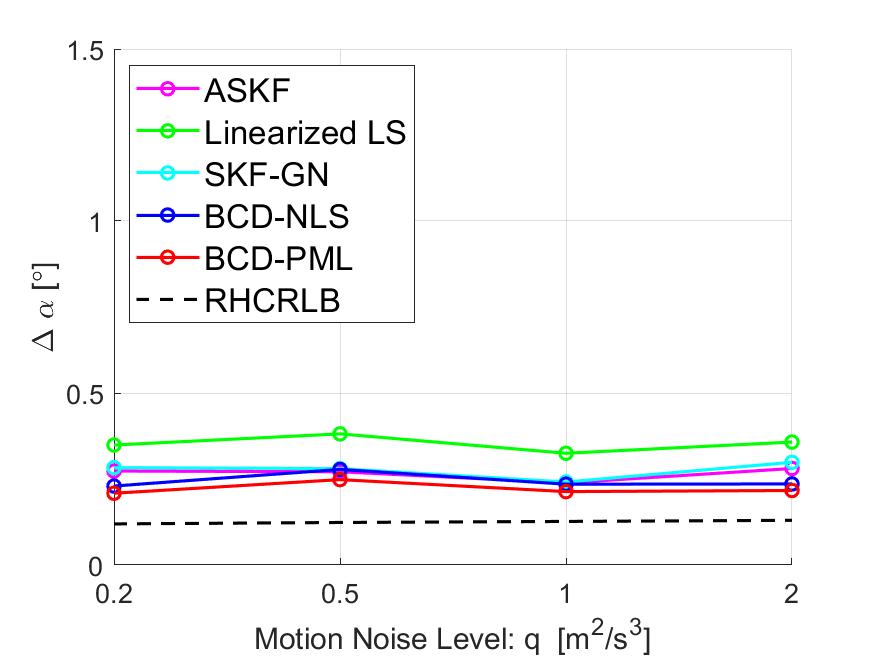}}
	\subfigure[Pitch bias.]{
		\includegraphics[width=0.23\linewidth]{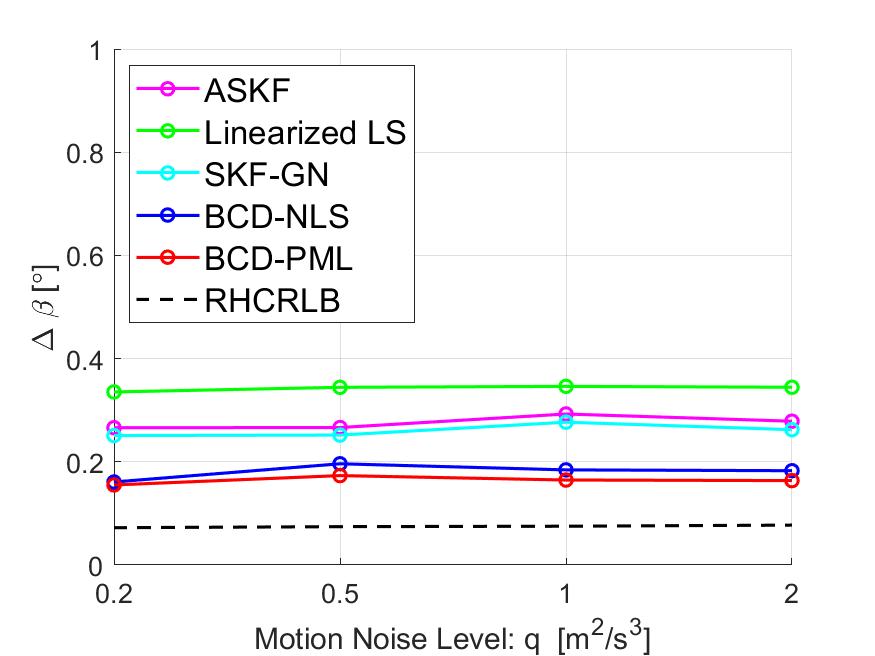}}
	\subfigure[Yaw bias.]{
		\includegraphics[width=0.23\linewidth]{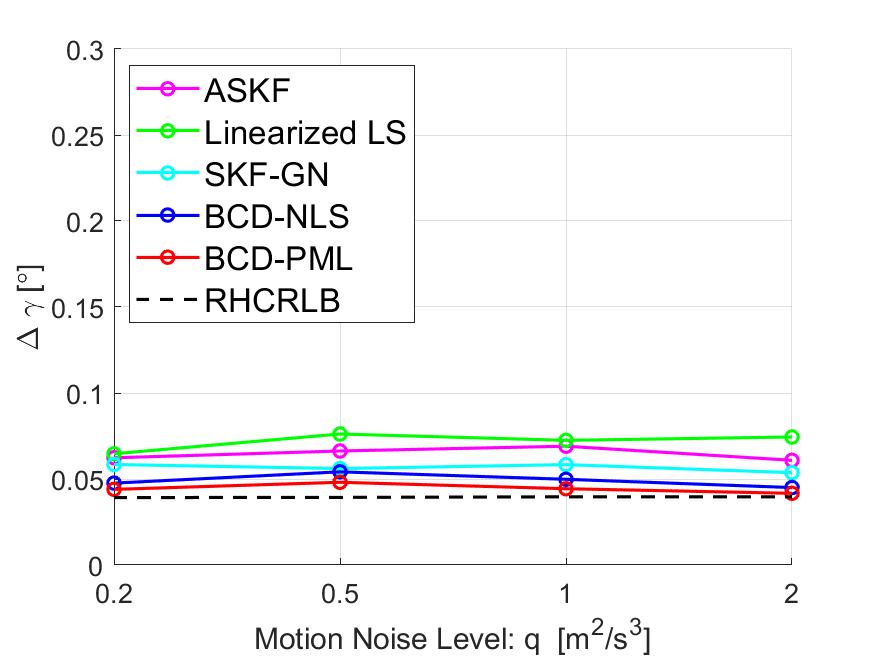}}
	\caption{RMSE and RHCRLB of rotation angle biases with different motion process noise levels $q$.}
	\label{fig:simul_mo2}
\end{figure*}

\begin{figure*}
	\centering 
	\subfigure[Elevation bias.]{
		\includegraphics[width=0.23\linewidth]{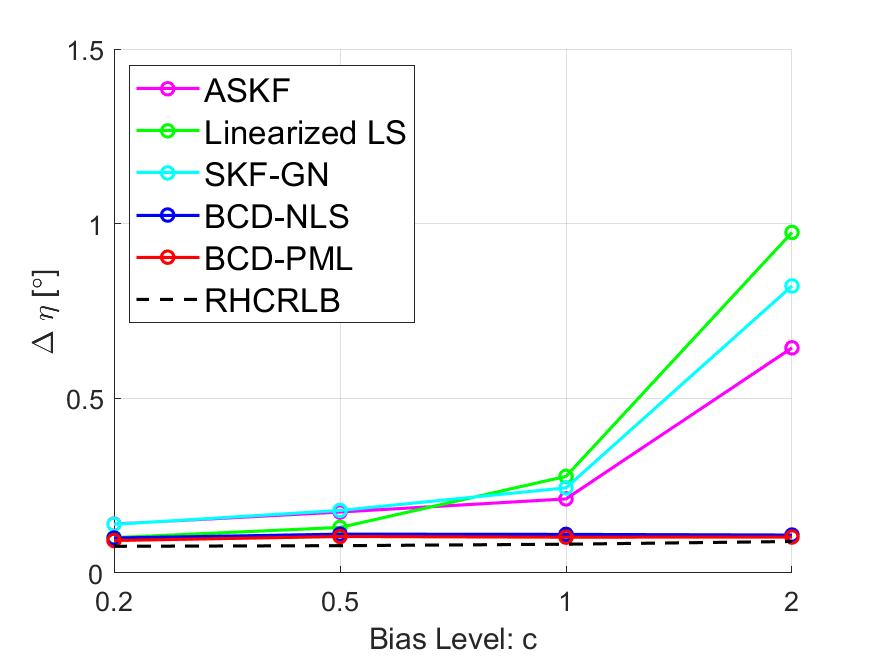}}
	\subfigure[Roll bias.]{
		\includegraphics[width=0.23\linewidth]{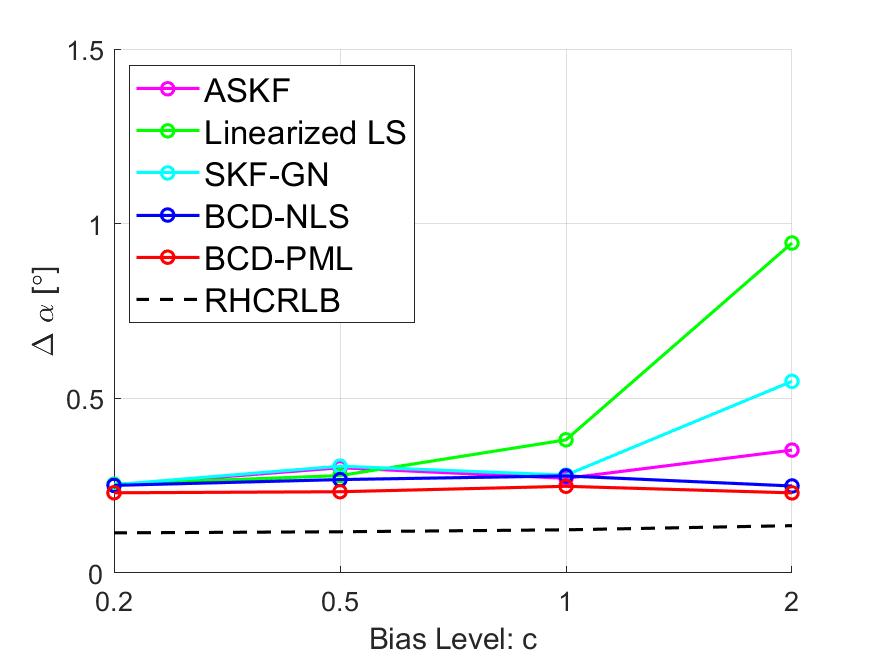}}
	\subfigure[Pitch bias.]{
		\includegraphics[width=0.23\linewidth]{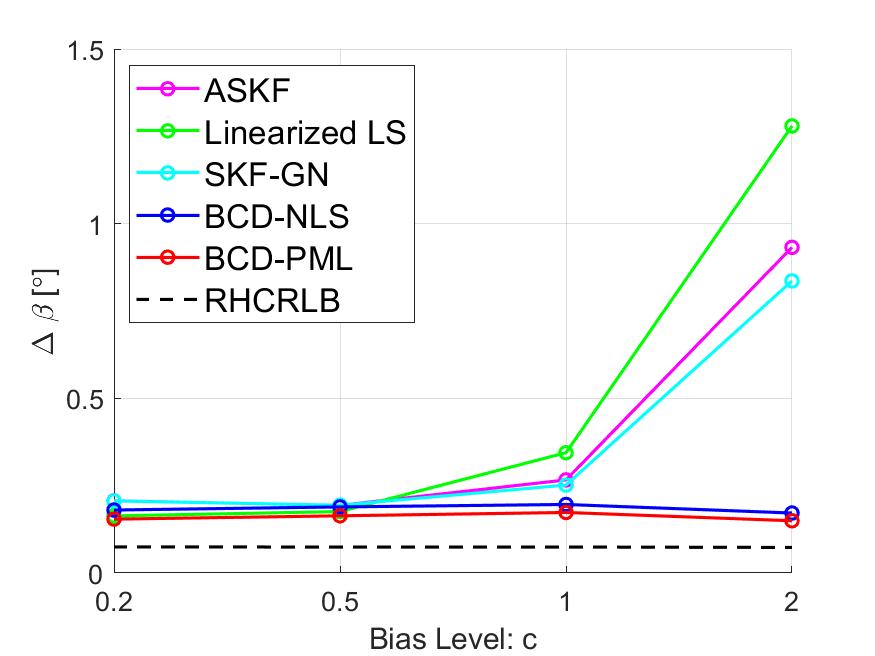}}
	\subfigure[Yaw bias.]{
		\includegraphics[width=0.23\linewidth]{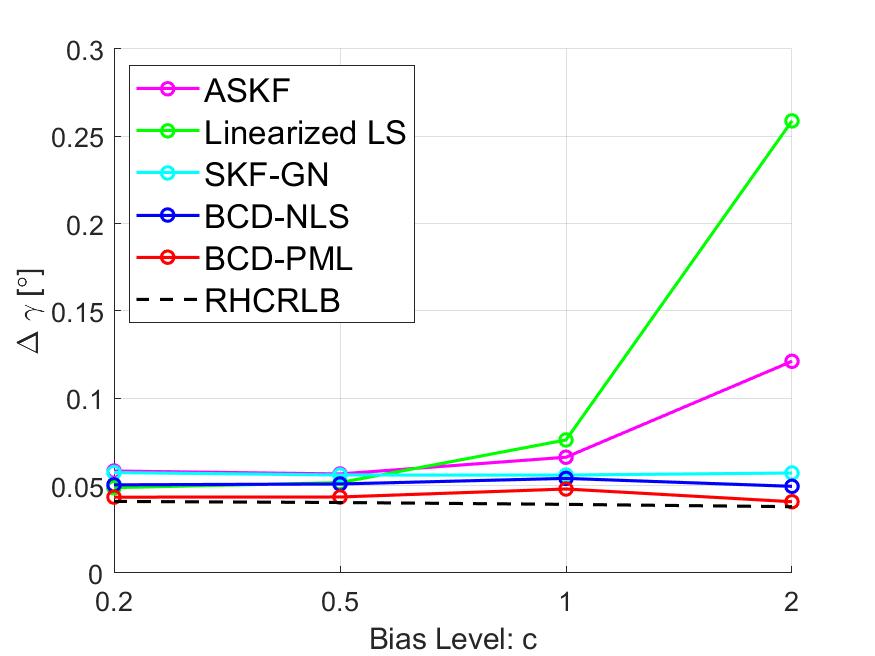}}
	\caption{RMSE and RHCRLB of rotation angle biases with different bias levels $c$.}
	\label{fig:simul_bias2}
\end{figure*}
\section{Conclusion and Future Work}
In this paper, we have presented a weighted nonlinear LS formulation for the 3-dimensional asynchronous multi-sensor registration problem by assuming the existence of a reference target moving with a nearly constant velocity. By choosing appropriate weight matrices, the proposed formulation includes classical nonlinear LS and ML estimation as special cases. Moreover, we have proposed an efficient BCD algorithm for solving the formulated problem and also have developed a computationally efficient algorithm based on the ADMM to globally solve nonconvex subproblems with respect to angle biases in the BCD iteration. Numerical simulation results demonstrate the effectiveness  and efficiency of the proposed formulation and the BCD algorithm. In particular, our proposed approaches can achieve significantly better estimation performance than existing celebrated approaches based on the linear approximation and our proposed approaches are the most robust to the sensor bias level.

It is important to note that the proposed formulation~\eqref{eq:opt_nls} implicitly assumes that all observed data are originated by a single moving target. Other important issues in practical multi-sensor systems such as missed detection, multiple targets, and clutters~\cite{bar1990tracking} are not considered. Such more complicated scenarios are of greatly practical interests, and dealing with the data association problem (with biased data) raised there is necessary~\cite{crouse2009sensor}. Combining the developed optimization technique together with advanced multi-target tracking techniques~\cite{meyer2018message} would be an interesting direction to further explore, especially for sensors on moving platforms. Finally, we remark that the use of the rotation matrix in~\eqref{eq: rotation} may lead to a gimbal lock problem when the pitch angle is $\pm 90^\circ$. However, in the context of this study, where sensors are fixed at a ground location with fixed positions and attitude angles, this issue can be avoided by pre-aligning the rotation angles between the local and global Cartesian coordinates being zeros. In contrast, for sensors mounted on moving platforms, it is advisable to use quaternions to prevent gimbal lock from occurring. Exploring techniques for compensating biases in quaternion representations would be an interesting future research.

\appendices
\section{Proof of Proposition~\ref{prop:amg}}\label{app:propamg}
Since the result is independent with respect to noise $\bm{w}_k$, we assume $\bm{w}_k=\bm{0}$ for notational simplicity. Define $(r,\omega,\upsilon)$ such that  $[r\cos\omega\cos\upsilon,r\sin\omega\cos\upsilon,r\sin\upsilon]^T=R_y^T(\beta_{s_k}+\Delta\beta_{s_k})R_x^T(\alpha_{s_k}+\Delta\alpha_{s_{k}})(\bm{\xi}_k-\bm{p}_{s_k})$. Then, by  \eqref{eq:mea_model} and omitting the subscripts $k$ and $s_k$, we have
	\begin{align*}
	\small 
	 \begin{bmatrix} {\rho}\\ {\phi} \\ {\eta} \end{bmatrix}&=h\left(R_z^T(\gamma+\Delta\gamma)\begin{bmatrix} r\cos\omega\cos\upsilon\\r\sin\omega\cos\upsilon\\r\sin\upsilon\end{bmatrix}\right)-\begin{bmatrix} \Delta\rho\\ \Delta\phi \\ \Delta\eta\end{bmatrix}\\
	 &=h\left( \begin{bmatrix} r\cos(\omega-\gamma-\Delta\gamma)\cos\upsilon\\r\sin(\omega-\gamma-\Delta\gamma)\cos\upsilon\\r\sin\upsilon\end{bmatrix}\right)-\begin{bmatrix} \Delta\rho\\ \Delta\phi \\ \Delta\eta\end{bmatrix}\\
	 &=\begin{bmatrix} r-\Delta\rho\\ \omega-\gamma-\Delta\gamma-\Delta\phi \\ \upsilon-\Delta \eta\end{bmatrix}.
	\end{align*}
This completes the proof.

\section{Proof of Theorem \ref{thm:admm}}\label{app:thm}

In this appendix, we use notation $\Delta\bar{\bm{\vartheta}}$ to denote the true value of the corresponding bias and define $\bm{x}(\Delta{\bm{\vartheta}})= [\cos\Delta {\vartheta}_1, \sin\Delta {\vartheta}_1 ,\dots ,\cos\Delta {\vartheta}_M, \sin\Delta {\vartheta}_M]^T$. Also, for notational simplicity, we ignore the subscript $\vartheta$ of $\bm{H}_\vartheta$ and $\bm{c}_\vartheta$. Since the objective function in~\eqref{eq:qcqp} is a quadratic function whose gradient is Lipschitz continuous and the constraint set $\mathcal{C}$ is a smooth compact manifold, then by~\cite[Corollary 2] {wang2019global}, for a sufficiently large penalty parameter $\rho$, the sequence $\left\{\bm{x}^t,\bm{z}^t,\bm{\lambda}^t\right\}$ generated by Algorithm~\ref{alg:admm} has at least one limit point and any limit point is a stationary point of the augmented Lagrangian function $L_\rho(\bm{x},\bm{z},\bm{\lambda})$. This completes the first statement in Theorem~\ref{thm:admm}.

Next, we prove the second statement. From the definition of $L_\rho(\bm{x},\bm{z},\bm{\lambda})$ in \eqref{eq:Lrho}, its stationary point $(\bm{x},\bm{z},\bm{\lambda})$ is defined as 
\begin{subequations}\label{appeq:sta}
\begin{align}
&\frac{\partial L_\rho}{\partial \bm{x}}=2\bm{H}^T\bm{Q}\bm{H}\bm{x}+2\bm{H}^T\bm{Q}\bm{c}+\bm{\lambda}+\rho\bm{x}-\rho\bm{z}=\bm{0},\\
&\bm{0}\in \frac{\partial L_\rho}{\partial \bm{z}}=-\bm{\lambda}+\rho\bm{z}-\rho\bm{x}+\partial \iota_{\mathcal{C}}(\bm{z}),\\
&\frac{\partial L_\rho}{\partial \bm{\lambda}}=\bm{x}-\bm{z}=\bm{0}.
\end{align}
\end{subequations}
In the above, $\partial \iota_{\mathcal{C}}(\bm{z})$ denotes the Fr\'{e}chet  subdifferential~\cite{kruger2003frechet}, and by Proposition \ref{prop:subg} below, we have $\partial \iota_{\mathcal{C}}(\bm{z})=\bm{\Lambda}\bm{z}$, where $\bm{\Lambda}=\textrm{Diag}(\bm{\eta})\otimes \bm{I}_2$ and $\bm{\eta}\in\mathbb{R}^M$.
Then, \eqref{appeq:sta} can be simplified as follows:
\begin{equation}\label{appeq:spKKT}
\bm{A}\bm{x}+\bm{b}+\bm{\Lambda}\bm{x}=\bm{0},\ \bm{x}\in\mathcal{C}.
\end{equation}
where $\bm{A}=2\bm{H}^T\bm{Q}\bm{H}$ and $\bm{b}=2\bm{H}^T\bm{Q}\bm{c}$. By \eqref{eq:decomp}, we have $\bm{H}=\bar{\bm{H}}+\Delta\bm{H}$ and $\bm{c}=\bar{\bm{c}}+\Delta\bm{c}$.

Express $\bm{x}$ by $\Delta \bm{\vartheta}$, i.e., $\bm{x}=\bm{x}(\Delta \bm{\vartheta})$. Then the constraint $\bm{x}\in\mathcal{C}$ in~\eqref{appeq:spKKT} can be eliminated and the left-hand side of \eqref{appeq:spKKT} can be regarded as a function with respect to $\Delta\bm{H}$, $\Delta\bm{c}$, $\Delta\bm{\vartheta}$, and $\bm{\eta}$, denoted as $f(\Delta \bm{H},\Delta \bm{c},\Delta\bm{\vartheta},\bm{\eta})$; see~\eqref{appeq:f} further ahead. Thus,~\eqref{appeq:spKKT} is further re-expressed as 
\begin{equation*}
f(\Delta \bm{H},\Delta \bm{c},\Delta\bm{\vartheta},\bm{\eta})=\bm{0}.
\end{equation*}
Then, by Proposition \ref{prop:imp} below, there exist two continuous functions $d_{\vartheta}(\Delta \bm{H},\Delta \bm{c})$ and $d_{\eta}(\Delta \bm{H},\Delta \bm{c})$ in the neighborhood of $(\Delta \bm{H},\Delta \bm{c})=(\bm{0},\bm{0})$ and these two functions are uniquely defined in this neighborhood such that 
$$f(\Delta \bm{H},\Delta \bm{c},\Delta\tilde{\bm{\vartheta}},\tilde{\bm{\eta}})=\bm{0}$$ 
holds with $\Delta\tilde{\bm{\vartheta}}=d_{\vartheta}(\Delta \bm{H},\Delta \bm{c})$ and $\tilde{\bm{\eta}}=d_{\eta}(\Delta \bm{H},\Delta \bm{c})$. The uniqueness of these two functions implies that, for any sufficiently small $(\Delta \bm{H},\Delta \bm{c})$, $(\Delta\tilde{\bm{\vartheta}},\tilde{\bm{\eta}})$ is the only solution (depending on $(\Delta \bm{H},\Delta \bm{c})$) such that \eqref{appeq:spKKT} holds. Hence, by the equivalence between~\eqref{appeq:sta} and~\eqref{appeq:spKKT}, $$(\bar{\bm{x}},\bar{\bm{z}},\bar{\bm{\lambda}})=\left(\bm{x}(\Delta \tilde{\bm{\vartheta}}),\bm{x}(\Delta \tilde{\bm{\vartheta}}),[\textrm{Diag}(\tilde{\bm{\eta}})\otimes \bm{I}_2]\bm{x}(\Delta \tilde{\bm{\vartheta}})\right)$$ 
is the unique stationary point of $L_\rho(\bm{x},\bm{z},\bm{\lambda})$ (for sufficiently small $(\Delta \bm{H},\Delta \bm{c})$). Furthermore, combining with \cite[Corollary 2]{wang2019global}, we know that $(\bar{\bm{x}},\bar{\bm{z}},\bar{\bm{\lambda}})$ is the unique limit point of the sequence $\{\bm{x}^t,\bm{z}^t,\bm{\lambda}^t\}$. By~\cite[Lemma 6]{wang2019global}, $\{\bm{x}^t,\bm{z}^t,\bm{\lambda}^t\}$ is bounded, and hence every subsequence of $\{\bm{x}^t,\bm{z}^t,\bm{\lambda}^t\}$ converges $(\bar{\bm{x}},\bar{\bm{z}},\bar{\bm{\lambda}}).$ Otherwise, there exists some subsequence of $\{\bm{x}^t,\bm{z}^t,\bm{\lambda}^t\}$ converging to a point which is different from $(\bar{\bm{x}},\bar{\bm{z}},\bar{\bm{\lambda}})$ and this contradicts the fact that $(\bar{\bm{x}},\bar{\bm{z}},\bar{\bm{\lambda}})$ is the unique limit point. As such, the entire sequence $\{ \bm{x}^t,\bm{z}^t,\bm{\lambda}^t\}$ converges to $(\bar{\bm{x}},\bar{\bm{z}},\bar{\bm{\lambda}})$.

Finally, if $\Delta\bm{H}=\bm{0}$ and $\Delta \bm{c}=\bm{0}$, $f(\bm{0},\bm{0},\Delta\bar{\bm{\vartheta}},\bm{0})$ corresponds to the gradient of $\| \bar{\bm{H}} \bm{x} + \bar{\bm{c}}\|_{\bm{Q}}^2$ with respect to $\bm{x}$ at point $\bm{x}(\Delta \bar{\bm{\vartheta}})$. Then, by Proposition~\ref{prop:delta}, it is straightforward to check that $f(\bm{0},\bm{0},\Delta\bar{\bm{\vartheta}},\bm{0})=\bm{0}$, where $\Delta\bar{\bm{\vartheta}}$ corresponds to the true bias. Therefore, the unique stationary point $\bar{\bm{x}}$ is $\bm{x}(\Delta\bar{\bm{\vartheta}})$. This completes the proof the second statement of Theorem 1.

\begin{prop}\label{prop:subg}
For the indicator function $\iota_{\mathcal{C}}(\cdot)$ defined in \eqref{eq:indfunc}, its Fr\'{e}chet subdifferential at point $\bm{z}$ is 
$$
\partial \iota_{\mathcal{C}}(\bm{z})=\left\{  \bm{\Lambda}\bm{z} \mid \bm{\Lambda}=\textrm{Diag}(\bm{\eta})\otimes \bm{I}_2, \bm{\eta}\in\mathbb{R}^M \right\}.$$
\end{prop}\begin{proof}
	We first reformulate the indicator function  $\iota_{\mathcal{C}}(\cdot)$ as $\iota_{\mathcal{C}}(\bm{z})=\sum_{m=1}^M\iota(\bm{z}_{2m-1},\bm{z}_{2m}),$ where 
	\begin{equation}\label{eq:app_iota}
	\iota(\bm{u})=\left\lbrace
	\begin{aligned}
		0,\ &\textrm{if } \bm{u}\in\mathcal{U}\triangleq\{\bm{u}\in\mathbb{R}^2\mid   u_1^2+u_2^2=1 \},\\
		\infty, \ &\textrm{otherwise}.
	\end{aligned}
	\right.
	\end{equation}
Then, the Fr\'{e}chet subdifferential $\partial \iota(\cdot)$ of  $\iota(\cdot)$ at point $\bm{u}\in\mathbb{R}^{2}$ is the Fr\'{e}chet normal cone of set $\mathcal{U}$~\cite[Proposition 1.18]{kruger2003frechet}, which is defined as 
\begin{align*}
&C\iota(\bm{u})=
\left\{\bm{g}\in\mathbb{R}^{2}\mid\liminf_{\bm{u}^\prime\stackrel{\mathcal{U}}\rightarrow\bm{u}}\frac{-\bm{g}^T(\bm{u}^\prime-\bm{u})}{\|\bm{u}^\prime-\bm{u} \|}\geq 0\right\},
\end{align*}	
where the notation $\bm{u}^\prime\stackrel{\mathcal{U}}\rightarrow\bm{u}$ represents $\bm{u}^\prime\rightarrow\bm{u}$ with $\bm{u},\bm{u}^\prime\in \mathcal{U}$. By expressing $\bm{u}=(\cos\theta,\sin\theta)$ and $\bm{u}^\prime=(\cos\theta^\prime,\sin\theta^\prime)$ with $\theta,\theta^\prime\in[-\pi,\pi]$, we have 
    \begin{equation*}
    \begin{aligned}
        \lim_{\bm{u}^\prime\stackrel{\mathcal{U}}\rightarrow\bm{u}}\frac{\bm{u}^\prime-\bm{u}}{\|\bm{u}^\prime-\bm{u} \|}&=\lim_{\theta^\prime\rightarrow\theta}\frac{(\cos\theta^\prime-\cos\theta,\sin\theta^\prime-\sin\theta)}{\sqrt{(\cos\theta^\prime-\cos\theta)^2+(\sin\theta^\prime-\sin\theta)^2}}\\
        &\overset{(i)}{=}\lim_{\theta^\prime\rightarrow\theta}\frac{(\cos\theta^\prime,\sin\theta^\prime)-(\cos\theta,\sin\theta)}{\sqrt{2-2\cos(\theta^\prime-\theta)}}\\
        &\overset{(ii)}{=}\left\{ 
        \begin{aligned}
         \lim_{\theta^\prime\rightarrow\theta^+}\frac{(\cos\theta^\prime-\cos\theta,\sin\theta^\prime-\sin\theta)}{2\sin(\frac{\theta^\prime-\theta}{2})}\\
         \lim_{\theta^\prime\rightarrow\theta^-}\frac{(\cos\theta^\prime-\cos\theta,\sin\theta^\prime-\sin\theta)}{-2\sin(\frac{\theta^\prime-\theta}{2})}
        \end{aligned}\right.\\
        &\overset{(iii)}{=}(\mp\sin\theta,\pm\cos\theta)=\pm\begin{bmatrix} 0 & -1 \\ 1 & 0 \end{bmatrix}\bm{u},
    \end{aligned}
\end{equation*}
where $\theta^\prime\rightarrow\theta^+$ (or $\theta^-$) denotes that $\theta^\prime$ approaches $\theta$ from $\theta^\prime>\theta$ (or $\theta^\prime<\theta$). In the above, (i) is due to $\cos(\alpha-\beta)=\cos\alpha\cos\beta+\sin\alpha\sin\beta$, (ii) is by $\sqrt{2}\sin\frac{\alpha}{2}=\pm\sqrt{1-\cos \alpha}$, and (iii) is by the L'H\"{o}pital's rule. Hence, 
\begin{align*}
	&\liminf_{\bm{u}^\prime\stackrel{\mathcal{U}}\rightarrow\bm{u}}\frac{-\bm{g}^T(\bm{u}^\prime-\bm{u})}{\|\bm{u}^\prime-\bm{u} \|}\geq 0\Leftrightarrow\lim_{\bm{u}^\prime\stackrel{\mathcal{U}}\rightarrow\bm{u}}\frac{-\bm{g}^T(\bm{u}^\prime-\bm{u})}{\|\bm{u}^\prime-\bm{u} \|}= 0\\
	&\Leftrightarrow\bm{g}^T\begin{bmatrix} 0 & -1 \\ 1 & 0 \end{bmatrix}\bm{u}= 0\Leftrightarrow \bm{g}=\eta \bm{u},\eta\in\mathbb{R},
\end{align*}
and $\partial\iota(\bm{u})=\{\eta\bm{u}\in\mathbb{R}^2\mid\eta\in\mathbb{R}\}.$

For $\iota_{\mathcal{C}}(\bm{z})=\sum_{m=1}^M\iota(\bm{z}_{2m-1},\bm{z}_{2m})$, denoting 
$$\bm{u}_m=[\bm{z}_{2m-1},\bm{z}_{2m}]^T,m=1,2,\ldots,M,$$
we have 
\begin{align*}
\partial\iota_{\mathcal{C}}(\bm{z})=&(\partial\iota(\bm{u}_1),\partial\iota(\bm{u}_2),\ldots,\partial\iota(\bm{u}_M))\\
=&\left\{  \bm{\Lambda}\bm{z} \mid \bm{\Lambda}=\textrm{Diag}(\bm{\eta})\otimes \bm{I}_2, \bm{\eta}\in\mathbb{R}^M \right\},
\end{align*}
where $\otimes$ is the Kronecker product.
\end{proof}

\begin{prop}\label{prop:imp}
	 Express $\bm{x}$ by $\Delta \bm{\vartheta}$, i.e., $\bm{x}=\bm{x}(\Delta \bm{\vartheta})$, then $\eqref{appeq:spKKT}$ can be expressed as $f(\Delta \bm{H},\Delta \bm{c},\Delta\bm{\vartheta},\bm{\eta})=\bm{0}$, where
	\begin{equation}\label{appeq:f}
	\begin{aligned}
	f(\Delta \bm{H},\Delta \bm{c},\Delta\bm{\vartheta},\bm{\eta})
	=&2(\bar{\bm{H}}+\Delta\bm{H})^T\bm{Q}(\bar{\bm{H}}+\Delta\bm{H})\bm{x}(\Delta \bm{\vartheta})\\
	&+2(\bar{\bm{H}}+\Delta\bm{H})^T\bm{Q}(\bar{\bm{c}}+\Delta\bm{c})\\
	&+[\textrm{Diag}(\bm{\eta})\otimes \bm{I}_2]\bm{x}(\Delta \bm{\vartheta}).
	\end{aligned}
	\end{equation}
	Further, there exist two continuous functions $d_{\vartheta}(\Delta \bm{H},\Delta \bm{c})$ and $d_{\eta}(\Delta \bm{H},\Delta \bm{c})$ in the neighborhood of $(\Delta \bm{H},\Delta \bm{c})=(\bm{0},\bm{0})$ and these two functions are uniquely defined in this neighborhood such that $f(\Delta \bm{H},\Delta \bm{c},d_{\vartheta}(\Delta \bm{H},\Delta \bm{c}),d_{\eta}(\Delta \bm{H},\Delta \bm{c}))=\bm{0}$. 
\end{prop}

\begin{proof}
By \eqref{eq:decomp}, we have $\bm{H}=\bar{\bm{H}}+\Delta\bm{H}$ and $\bm{c}=\bar{\bm{c}}+\Delta\bm{c}$. From Proposition \ref{prop:delta} below, we know that \eqref{appeq:spKKT} holds for $\Delta\bm{H}=\bm{0},\Delta\bm{c}=\bm{0},\bm{\eta}=\bm{0}$, and $\bm{x}=\bm{x}(\Delta\bar{\bm{\vartheta}})$. This is equivalent to $f(\bm{0},\bm{0},\Delta\bar{\bm{\vartheta}},\bm{0})=\bm{0}$. Moreover, by \cite[Theorem 9.3] {loomis1968advanced}, if the Jacobian matrix $\bm{D}_{\Delta\bm{\vartheta},\bm{\eta}}$ of $f(\Delta \bm{H},\Delta \bm{c},\Delta\bm{\vartheta},\bm{\eta})$ with respect to $\Delta\bm{\vartheta}$ and $\bm{\eta}$ is invertible at point $(\Delta \bm{H},\Delta \bm{c},\Delta\bm{\vartheta},\bm{\eta})=(\bm{0},\bm{0},\Delta\bar{\bm{\vartheta}},\bm{0})$, then Proposition \ref{prop:imp} holds. Next, we show that the Jacobian matrix $\bm{D}_{\Delta\bm{\vartheta},\bm{\eta}}$ at $(\bm{0},\bm{0},\Delta\bar{\bm{\vartheta}},\bm{0})$ is indeed invertible. 
	
By \eqref{appeq:f}, $\bm{D}_{\Delta\bm{\vartheta},\bm{\eta}}$ is given as 
\begin{equation*}
    \bm{D}_{\Delta\bm{\vartheta},\bm{\eta}}=\left[ \frac{\partial f}{\partial \bm{\eta}}\ \frac{\partial f}{\partial \Delta\bm{\vartheta}}\right]=\left[  \bm{S}\ \bm{G}\bm{C}\right]=\begin{bmatrix}  \bm{I}_{2M}&\bm{G}\end{bmatrix}\begin{bmatrix} \bm{S}&\bm{0}\\ \bm{0} & \bm{C} \end{bmatrix},
\end{equation*}
where $$\bm{G}=\bm{H}^T\bm{Q}\bm{H}+[\textrm{Diag}(\bm{\eta})\otimes \bm{I}_2],$$
	$$ \bm{S}=\textrm{blkdiag}(\bm{s}_1,\bm{s}_2,\ldots,\bm{s}_M),\bm{s}_m=[\cos \Delta{\vartheta}_m,\sin \Delta{\vartheta}_m]^T,$$
	$$ \bm{C}=\textrm{blkdiag}(\bm{c}_1,\bm{c}_2,\ldots,\bm{c}_M),\bm{c}_m=[-\sin \Delta{\vartheta}_m,\cos \Delta{\vartheta}_m]^T.$$
	Hence, at point $(\Delta \bm{H},\Delta \bm{c},\Delta\bm{\vartheta},\bm{\eta})=(\bm{0},\bm{0},\Delta\bar{\bm{\vartheta}},\bm{0})$, we have 
	$$\textrm{rank}\left(\begin{bmatrix}  \bm{I}_{2M}&\bm{G}\end{bmatrix}\right)=2M, \ \textrm{rank}\left(\begin{bmatrix} \bm{S}&\bm{0}\\ \bm{0} & \bm{C} \end{bmatrix}\right)=2M,$$
which, together with Sylvester's rank inequality \cite{Horn1985matrix}, immediately implies $\textrm{rank}(\bm{D}_{\Delta\bm{\vartheta},\bm{\eta}})=2M$ and thus  $\bm{D}_{\Delta\bm{\vartheta},\bm{\eta}}$ is invertible. This completes the proof.
\end{proof}

\begin{prop}\label{prop:delta}
If $\Delta\bm{H}_\vartheta=\bm{0}$ and $\Delta\bm{c}_\vartheta=\bm{0}$, then we have ${\bm{H}}_\vartheta {\bm{x}}(\Delta\bar{\bm{\vartheta}})+{\bm{c}}_\vartheta=\bm{0}$, where $\Delta\bar{\bm{\vartheta}}$ is the true value of the corresponding bias.
\end{prop}
\begin{proof}
According to the definitions of $\Delta\bm{H}_\vartheta$ and $\Delta\bm{c}_\vartheta$ and the derivations in Appendix~\ref{app:deriv}, we know that $\Delta\bm{H}_\vartheta=\bm{0}$ and $\Delta\bm{c}_\vartheta=\bm{0}$ hold when $\bm{w}_k=\bm{0},\bm{n}_k=\dot{\bm{n}}_k=\bm{0}$, for all $ k$, $\lambda_\eta=\lambda_\phi=1$, and all the other biases $\Omega\setminus \Delta\bm{\vartheta}$ and velocity $\bm{v}$ take the true values. In this case, the constructed measurement equations in~\eqref{eq:bias_mea_model} become
\begin{equation*}
\begin{aligned}
\bm{0}&=g_{k+1}(  \boldsymbol{\theta}_{s_{k+1}}) - 
g_{k}( \boldsymbol{\theta}_{s_k}) - T_k\dot{\bm{\xi}}_k,\ \forall\,k,\\
\bm{0}&= \dot{\bm{\xi}}_{k+1} - \dot{\bm{\xi}}_k,\  \forall\,k,
\end{aligned}
\end{equation*}
where $\boldsymbol{\theta}_{m},m=1,2,\ldots, M$, take the true bias values. Since problem~\eqref{eq:qcqp} is an equivalent reformulation of problem~\eqref{eq:opt_nls} with respect to one type of bias $\Delta \bm{\vartheta}$, then the above equation can also be reformulated with respect to $\Delta \bm{\vartheta}$ as below 
$${\bm{H}}_\vartheta {\bm{x}}(\Delta\bar{\bm{\vartheta}})+{\bm{c}}_\vartheta=\bm{0},$$
where $\Delta\bar{\bm{\vartheta}}$ corresponds to the true bias value. This completes the proof.
\end{proof}

\bibliographystyle{IEEEtran}
\bibliography{refs}

\vfill\pagebreak
\newpage
\clearpage
\begin{center} 
{\large\textbf{Supplementary Material}}
\end{center}
\normalsize
\section{Derivations of Subproblems \eqref{eq:sub_xi}--\eqref{eq:sub_gamma}}\label{app:deriv}
Denote 
$\bm{Q}=\textrm{blkdiag}(\bm{Q}_1,\bm{Q}_2,\ldots,\bm{Q}_{K}),$ where $\textrm{blkdiag}(\cdot)$ is  the diagonalizing operation for stacking matrices $\bm{Q}_1,\ldots,\bm{Q}_{K}$ as a block diagonal matrix. Also, for $\Delta{\bm{\vartheta}}\in\mathbb{R}^M$, define vector $\bm{x}(\Delta{\bm{\vartheta}})= [\cos\Delta {\vartheta}_1, \sin\Delta {\vartheta}_1 ,\dots ,\cos\Delta {\vartheta}_M, \sin\Delta {\vartheta}_M]^T$.

\subsection{Derivation for Subproblem \eqref{eq:sub_xi}}
With all biases being fixed, problem \eqref{eq:opt_nls} with respect to velocities is reformulated as follows:
\begin{equation}\label{eq:app_ref_v}
\min \limits_{\{  \dot{\bm{\xi}}_k \}}
\sum_{k=1}^{K-1}
\left\|  \bm{c}_{v,k} +\bm{H}_{v,k}\begin{bmatrix}  \dot{\bm{\xi}}_k \\ \dot{\bm{\xi}}_{k+1}\end{bmatrix}\right\|_{\bm{Q}_k}^2,
\end{equation}
where 
$$\bm{c}_{v,k}=\begin{bmatrix} g_{k+1}(  \boldsymbol{\theta}_{s_{k+1}}) - g_{k}( \boldsymbol{\theta}_{s_k}) \\ \bm{0} \end{bmatrix},\ \bm{H}_{v,k}=\begin{bmatrix}-T_k\bm{I}_3 & \bm{0} \\ -\bm{I}_3 & \bm{I}_3 \end{bmatrix}.$$ Compactly, defining $\bm{v}=[\dot{\bm{\xi}}_1,\dot{\bm{\xi}}_2,\ldots,\dot{\bm{\xi}}_{K}]^T$, we then have an equivalent form for \eqref{eq:app_ref_v} as follows: 
 \begin{equation}
 \min \limits_{\bm{v}}
 \| \bm{H}_v\bm{v} + \bm{c}_{v}\|_{\bm{Q}}^2,
 \end{equation}
 where $\bm{c}_{v}=[\bm{c}_{v,1}^T,\bm{c}_{v,2}^T,\ldots,\bm{c}_{v,K-1}^T]^T$ and $\bm{H}_{v}$ is a matrix with 
 $[\bm{H}_{v}]_{6k-5:6k,3k-2:3k+3}=\bm{H}_{v,k},\ k=1,2,\ldots,K-1$, and all the other elements being zeros.

\subsection{Derivation for Subproblem \eqref{eq:sub_rho}}\label{app:sub_rho}
With velocities and all the other biases being fixed (except $\Delta\bm{\rho}$), problem \eqref{eq:opt_nls} with respect to $\Delta\bm{\rho}$ is reformulated as follows: 
\begin{equation}\label{eq:app_ref_rho}
\begin{aligned}
\min_{ \Delta\bm{\rho}}\ \sum_{k=1}^{K-1}\| &\bm{h}_{\rho,k+1}\Delta\rho_{s_{k+1}}-\bm{h}_{\rho,k}\Delta\rho_{s_{k}}\\
&+\bm{c}_{\rho,{k+1}}-\bm{c}_{\rho,k} +\bm{v}_{\rho,k} \|_{\bm{Q}_k}^2,
\end{aligned}
\end{equation}
where
\begin{align*}
\scriptsize 
\bm{h}_{\rho,k}=&[\bar{\bm{h}}_{\rho,k}^T,\bm{0}^T]^T,\ \bm{c}_{\rho,k}=[\bar{\bm{c}}_{\rho,k}^T,\bm{0}^T]^T,\\
\bm{v}_{\rho,k}=&[\bm{p}_{s_{k+1}}^T- \bm{p}_{s_{k}}^T-T_k\dot{\bm{\xi}}_k^T,\dot{\bm{\xi}}_{k+1}^T-\dot{\bm{\xi}}_k^T]^T,\\
\bar{\bm{h}}_{\rho,k}=R(\bm{\zeta}_{s_k}+&\Delta \bm{\zeta}_{s_k})\left[\begin{matrix}
\lambda_\phi^{-1} \lambda_\eta^{-1}  \cos\phi_k \cos(\eta_k+\Delta\eta_{s_k})    \\
\lambda_\phi^{-1} \lambda_\eta^{-1} \sin\phi_k \cos(\eta_k+\Delta\eta_{s_k})  \\
\lambda_\eta^{-1}   \sin(\eta_k+\Delta\eta_{s_k}) 
\end{matrix}\right],\\
\bar{\bm{c}}_{\rho,k}=R(\bm{\zeta}_{s_k}+&\Delta \bm{\zeta}_{s_k})\left[\begin{matrix}
\lambda_\phi^{-1} \lambda_\eta^{-1}  \rho_k\cos\phi_k \cos(\eta_k+\Delta\eta_{s_k})    \\
\lambda_\phi^{-1} \lambda_\eta^{-1} \rho_k\sin\phi_k \cos(\eta_k+\Delta\eta_{s_k})  \\
\lambda_\eta^{-1}  \rho_k  \sin(\eta_k+\Delta\eta_{s_k}) 
\end{matrix}\right].
\end{align*}
Compactly, \eqref{eq:app_ref_rho} can be equivalently formulated as 
\begin{equation}
\min \limits_{\Delta\bm{\rho}}\  \| \bm{H}_{\rho} \Delta\bm{\rho}
+ \bm{ c}_{{\rho}}\|_{\bm{Q}}^2,
\end{equation}
where $\bm{H}_{\rho}$ is a zero matrix except
$$
[\bm{H}_{\rho}]_{6k-5:6k,m}=\left\{ 
\begin{aligned}
&\bm{h}_{\rho,k+1},\ \textrm{if }m=s_{k+1},\\
&-\bm{h}_{\rho,k},\ \textrm{if }m=s_{k},
\end{aligned}\right.
$$
and $\bm{c}_\rho=[\tilde{\bm{c}}_{\rho,1}^T,\ldots,\tilde{\bm{c}}_{\rho,{K-1}}^T]^T$ with $\tilde{\bm{c}}_{\rho,k}=\bm{c}_{\rho,{k+1}}-\bm{c}_{\rho,k} + \bm{v}_{\rho,k} $.

\subsection{Derivation for Subproblems \eqref{eq:sub_eta}--\eqref{eq:sub_gamma}}\label{app:subdev}
For notational simplicity, we firstly define $\bm{u}_k$ as
$$
\bm{u}_k\triangleq\left[\begin{matrix}
\lambda_\phi^{-1} \lambda_\eta^{-1}  (\rho_k+\Delta\rho_{s_k})\cos\phi_k \cos(\eta_k+\Delta\eta_{s_k})    \\
\lambda_\phi^{-1} \lambda_\eta^{-1} (\rho_k+\Delta\rho_{s_k})\sin\phi_k \cos(\eta_k+\Delta\eta_{s_k})  \\
\lambda_\eta^{-1}  (\rho_k+\Delta\rho_{s_k})  \sin(\eta_k+\Delta\eta_{s_k}) 
\end{matrix}\right],
$$
which will be used in this subsection.
\subsubsection{Subproblem \eqref{eq:sub_eta}}
With velocities and all other kinds of biases being fixed (except $\Delta\bm{\eta}$), problem \eqref{eq:opt_nls} with respect to  $\Delta\bm{\eta}$ is reformulated as follows: 
\begin{equation}\label{eq:app_ref_eta}
\begin{aligned}
\min_{ \Delta\bm{\eta}}\ \sum_{k=1}^{K-1}\| &\bm{h}_{\eta,k+1}^c\cos\Delta\eta_{s_{k+1}}+ \bm{h}_{\eta,k+1}^s\sin\Delta\eta_{s_{k+1}}-\\
&\bm{h}_{\eta,k}^c\cos\Delta\eta_{s_{k}} -\bm{h}_{\eta,k}^s\sin\Delta\eta_{s_{k}} +\bm{c}_{\eta,k}\|_{\bm{Q}_k}^2,
\end{aligned}
\end{equation}
where
$$\bm{h}_{\eta,k}^c=[(\bar{\bm{h}}_{\eta,k}^{c})^T,\bm{0}^T]^T, \ \bm{h}_{\eta,k}^s=[(\bar{\bm{h}}_{\eta,k}^{s})^T,\bm{0}^T]^T,\ \bm{c}_{\eta,k}=\bm{v}_{\rho,k},$$
and 
\begin{align*}
\scriptsize
\bar{\bm{h}}_{\eta,k}^c=&R(\bm{\zeta}_{s_k}+\Delta \bm{\zeta}_{s_k})\left[\begin{matrix}
\lambda_\phi^{-1} \lambda_\eta^{-1}  (\rho_k+\Delta\rho_{s_k})\cos\phi_k \cos\eta_k    \\
\lambda_\phi^{-1} \lambda_\eta^{-1} (\rho_k+\Delta\rho_{s_k})\sin\phi_k \cos\eta_k  \\
\lambda_\eta^{-1}  (\rho_k+\Delta\rho_{s_k})  \sin\eta_k 
\end{matrix}\right],\\
\bar{\bm{h}}_{\eta,k}^s=&R(\bm{\zeta}_{s_k}+\Delta \bm{\zeta}_{s_k})\left[\begin{matrix}
-\lambda_\phi^{-1} \lambda_\eta^{-1}  (\rho_k+\Delta\rho_{s_k})\cos\phi_k \sin\eta_k    \\
-\lambda_\phi^{-1} \lambda_\eta^{-1} (\rho_k+\Delta\rho_{s_k})\sin\phi_k \sin\eta_k  \\
\lambda_\eta^{-1}  (\rho_k+\Delta\rho_{s_k})  \cos\eta_k
\end{matrix}\right].
\end{align*}
Letting $\bm{x}=\bm{x}(\Delta\bm{\eta})$, \eqref{eq:app_ref_eta} can be equivalently reformulated as 
\begin{equation*}
\begin{aligned}
\min_{ \bm{x} }\  \|\bm{H}_{\eta} \bm{x} + \bm{ c}_{\eta}\|_{\bm{Q}}^2,\quad
\textrm{s.t.}\ \bm{x} \in \mathcal{C},
\end{aligned} 
\end{equation*}
where $\bm{H}_{\eta}$ is a zero matrix except
$$
[\bm{H}_{\eta}]_{6k-5:6k,m}=\left\{ 
\begin{aligned}
&\bm{h}_{\eta,k+1}^c,\ \textrm{if }m=2s_{k+1}-1,\\
&\bm{h}_{\eta,k+1}^s,\ \textrm{if }m=2s_{k+1},\\
&-\bm{h}_{\eta,k}^c,\ \textrm{if }m=2s_{k}-1,\\
&-\bm{h}_{\eta,k}^s,\ \textrm{if }m=2s_{k},
\end{aligned}\right.
$$
and $\bm{c}_\eta=[{\bm{c}}_{\eta,1}^T,{\bm{c}}_{\eta,2}^T,\ldots,{\bm{c}}_{\eta,{K-1}}^T]^T$.

\subsubsection{Subproblem \eqref{eq:sub_alp}}
With velocities  and all other kinds of biases being fixed (except $\Delta\bm{\alpha}$), problem \eqref{eq:opt_nls} with respect to  $\Delta\bm{\alpha}$ is reformulated as follows: 
\begin{equation}\label{eq:app_ref_alp}
\begin{aligned}
\min_{ \Delta\bm{\alpha}}\ \sum_{k=1}^{K-1}&\| \bm{h}_{\alpha,k+1}^c\cos\Delta\alpha_{s_{k+1}}+ \bm{h}_{\alpha,k+1}^s\sin\Delta\alpha_{s_{k+1}}-\\
&\bm{h}_{\alpha,k}^c\cos\Delta\alpha_{s_{k}} -\bm{h}_{\alpha,k}^s\sin\Delta\alpha_{s_{k}} +\bm{c}_{\alpha,k}\|_{\bm{Q}_k}^2,
\end{aligned}
\end{equation}
where
$$\bm{h}_{\alpha,k}^c=[(\bar{\bm{h}}_{\alpha,k}^{c})^T,\bm{0}^T]^T, \ \bm{h}_{\alpha,k}^s=[(\bar{\bm{h}}_{\alpha,k}^{s})^T,\bm{0}^T]^T,$$
{ and $\bm{c}_{\alpha,k}=\bm{v}_{\rho,k}+[\bar{\bm{c}}_{\alpha,k}^T,\bm{0}^T]^T$. Define notations 
$$
E_1=\begin{bmatrix}0&0&0\\0&1&0\\0&0&1 \end{bmatrix},\ E_2=\begin{bmatrix}0&0&0\\0&0&-1\\0&1&0 \end{bmatrix},\ E_3=\begin{bmatrix}1&0&0\\0&0&0\\0&0&0 \end{bmatrix},
$$}
{then $\bar{\bm{h}}_{\alpha,k}^c$, $\bar{\bm{h}}_{\alpha,k}^s$, and $ \bar{\bm{c}}_{\alpha,k}$ are expressed as}
\begin{align*}
\bar{\bm{h}}_{\alpha,k}^c= E_1 R_x(\alpha_{s_{k}})R_y(\beta_{s_k}+\Delta \beta_{s_k}) R_z(\gamma_{s_k}+\Delta \gamma_{s_k})\bm{u}_k,\\
\bar{\bm{h}}_{\alpha,k}^s=E_2 R_x(\alpha_{s_{k}})R_y(\beta_{s_k}+\Delta \beta_{s_k})R_z(\gamma_{s_k}+\Delta \gamma_{s_k})\bm{u}_k,\\
{\bar{\bm{c}}_{\alpha,k}}=E_3 R_x(\alpha_{s_{k}})R_y(\beta_{s_k}+\Delta \beta_{s_k})R_z(\gamma_{s_k}+\Delta \gamma_{s_k})\bm{u}_k.
\end{align*}
Letting $\bm{x}=\bm{x}(\Delta\bm{\alpha})$, \eqref{eq:app_ref_alp} can be equivalently reformulated as 
\begin{equation*}
\begin{aligned}
\min_{ \bm{x} }\  \|\bm{H}_{\alpha} \bm{x} + \bm{ c}_{\alpha}\|_{\bm{Q}}^2,\quad
\textrm{s.t.}\ \bm{x} \in \mathcal{C},
\end{aligned} 
\end{equation*}
where $\bm{H}_{\alpha}$ is a zero matrix except
$$
[\bm{H}_{\alpha}]_{6k-5:6k,m}=\left\{ 
\begin{aligned}
&\bm{h}_{\alpha,k+1}^c,\ \textrm{if }m=2s_{k+1}-1,\\
&\bm{h}_{\alpha,k+1}^s,\ \textrm{if }m=2s_{k+1},\\
&-\bm{h}_{\alpha,k}^c,\ \textrm{if }m=2s_{k}-1,\\
&-\bm{h}_{\alpha,k}^s,\ \textrm{if }m=2s_{k},
\end{aligned}\right.
$$
and $\bm{c}_\alpha=[{\bm{c}}_{\alpha,1}^T,{\bm{c}}_{\alpha,2}^T,\ldots,{\bm{c}}_{\alpha,{K-1}}^T]^T$.
\subsubsection{Subproblem \eqref{eq:sub_beta}}
With velocities and all other kinds of biases being fixed (except $\Delta\bm{\beta}$), problem \eqref{eq:opt_nls} with respect to $\Delta\bm{\beta}$ is reformulated as follows: 
\begin{equation}\label{eq:app_ref_beta}
\begin{aligned}
\min_{ \Delta\bm{\beta}}\ \sum_{k=1}^{K-1}&\| \bm{h}_{\beta,k+1}^c\cos\Delta\beta_{s_{k+1}}+ \bm{h}_{\beta,k+1}^s\sin\Delta\beta_{s_{k+1}}-\\
&\bm{h}_{\beta,k}^c\cos\Delta\beta_{s_{k}} -\bm{h}_{\beta,k}^s\sin\Delta\beta_{s_{k}} +\bm{c}_{\beta,k}\|_{\bm{Q}_k}^2,
\end{aligned}
\end{equation}
where
$$\bm{h}_{\beta,k}^c=[(\bar{\bm{h}}_{\beta,k}^{c})^T,\bm{0}^T]^T,\ \bm{h}_{\beta,k}^s=[(\bar{\bm{h}}_{\beta,k}^{s})^T,\bm{0}^T]^T,$$
{and $\bm{c}_{\beta,k}=\bm{v}_{\rho,k}+[\bar{\bm{c}}_{\beta,k}^T,\bm{0}^T]^T$. Define notations 
$$
E_1=\begin{bmatrix}1&0&0\\0&0&0\\0&0&1 \end{bmatrix},\ E_2=\begin{bmatrix}0&0&1\\0&0&0\\-1&0&0 \end{bmatrix},\ E_3=\begin{bmatrix}0&0&0\\0&1&0\\0&0&0 \end{bmatrix} ,
$$}
{then $\bar{\bm{h}}_{\beta,k}^c$, $\bar{\bm{h}}_{\beta,k}^s$, and $ \bar{\bm{c}}_{\beta,k}$ are expressed as}
\begin{align*}
\bm{h}_{\beta,k}^c=&R_x(\alpha_{s_k}+\Delta \alpha_{s_k})E_1 R_y(\beta_{s_k})R_z(\gamma_{s_k}+\Delta \gamma_{s_k})\bm{u}_k,\\
\bm{h}_{\beta,k}^s=&R_x(\alpha_{s_k}+\Delta \alpha_{s_k})E_2 R_y(\beta_{s_k})R_z(\gamma_{s_k}+\Delta \gamma_{s_k})\bm{u}_k,\\
{\bar{\bm{c}}_{\beta,k}}=&R_x(\alpha_{s_k}+\Delta \alpha_{s_k})E_3R_y(\beta_{s_k}) R_z(\gamma_{s_k}+\Delta \gamma_{s_k})\bm{u}_k.
\end{align*}
Letting $\bm{x}=\bm{x}(\Delta\bm{\beta})$, \eqref{eq:app_ref_beta} can be equivalently reformulated as 
\begin{equation*}
\begin{aligned}
\min_{ \bm{x} }\  \|\bm{H}_{\beta} \bm{x} + \bm{ c}_{\beta}\|_{\bm{Q}}^2,\quad
\textrm{s.t.}\ \bm{x} \in \mathcal{C},
\end{aligned} 
\end{equation*}
where $\bm{H}_{\beta}$ is a zero matrix except
$$
[\bm{H}_{\beta}]_{6k-5:6k,m}=\left\{ 
\begin{aligned}
&\bm{h}_{\beta,k+1}^c,\ \textrm{if }m=2s_{k+1}-1,\\
&\bm{h}_{\beta,k+1}^s,\ \textrm{if }m=2s_{k+1},\\
&-\bm{h}_{\beta,k}^c,\ \textrm{if }m=2s_{k}-1,\\
&-\bm{h}_{\beta,k}^s,\ \textrm{if }m=2s_{k},
\end{aligned}\right.
$$
and $\bm{c}_\beta=[{\bm{c}}_{\beta,1}^T,{\bm{c}}_{\beta,2}^T,\ldots,{\bm{c}}_{\beta,{K-1}}^T]^T$.

\subsubsection{Subproblem \eqref{eq:sub_gamma}}
With velocities all other kinds of biases being fixed (except $\Delta\bm{\gamma}$), problem \eqref{eq:opt_nls} with respect to $\Delta\bm{\gamma}$ is reformulated as follows: 
\begin{equation}\label{eq:app_ref_gamma}
\begin{aligned}
\min_{ \Delta\bm{\gamma}}\ \sum_{k=1}^{K-1}&\| \bm{h}_{\gamma,k+1}^c\cos\Delta\gamma_{s_{k+1}}+ \bm{h}_{\gamma,k+1}^s\sin\Delta\gamma_{s_{k+1}}-\\
&\bm{h}_{\gamma,k}^c\cos\Delta\gamma_{s_{k}} -\bm{h}_{\gamma,k}^s\sin\Delta\gamma_{s_{k}} +\bm{c}_{\gamma,k} \|_{\bm{Q}_k}^2,
\end{aligned}
\end{equation}
where
$$\bm{h}_{\gamma,k}^c=[(\bar{\bm{h}}_{\gamma,k}^{c})^T,\bm{0}^T]^T,\ \bm{h}_{\gamma,k}^s=[(\bar{\bm{h}}_{\gamma,k}^{s})^T,\bm{0}^T]^T,$$
{ and $\bm{c}_{\gamma,k}=\bm{v}_{\rho,k}+[\bar{\bm{c}}_{\gamma,k}^T,\bm{0}^T]^T$. Define notations 
$$
E_1=\begin{bmatrix}1&0&0\\0&1&0\\0&0&0 \end{bmatrix},\ E_2=\begin{bmatrix}0&-1&0\\1&0&0\\0&0&0 \end{bmatrix},\ E_3=\begin{bmatrix}0&0&0\\0&0&0\\0&0&1 \end{bmatrix},
$$}
{ then $\bar{\bm{h}}_{\gamma,k}^c$, $\bar{\bm{h}}_{\gamma,k}^s$, and $ \bar{\bm{c}}_{\gamma,k}$ are expressed as}
\begin{align*}
\bar{\bm{h}}_{\gamma,k}^c=&R_x(\alpha_{s_k}+\Delta \alpha_{s_k}) R_y(\beta_{s_k}+\Delta \beta_{s_k}) E_1R_z(\gamma_{s_k})\bm{u}_k,\\
\bar{\bm{h}}_{\gamma,k}^s=&R_x(\alpha_{s_k}+\Delta \alpha_{s_k}) R_y(\beta_{s_k}+\Delta \beta_{s_k}) E_2R_z(\gamma_{s_k})\bm{u}_k,\\
{\bar{\bm{c}}_{\gamma,k}}=&R_x(\alpha_{s_k}+\Delta \alpha_{s_k}) R_y(\beta_{s_k}+\Delta \beta_{s_k}) E_3R_z(\gamma_{s_k})\bm{u}_k.
\end{align*}
Letting $\bm{x}=\bm{x}(\Delta\bm{\gamma})$, \eqref{eq:app_ref_gamma} can be equivalently reformulated as 
\begin{equation*}
\begin{aligned}
\min_{ \bm{x} }\  \|\bm{H}_{\gamma} \bm{x} + \bm{ c}_{\gamma}\|_{\bm{Q}}^2,\quad
\textrm{s.t.}\ \bm{x} \in \mathcal{C},
\end{aligned} 
\end{equation*}
where $\bm{H}_{\gamma}$ is a zero matrix except
$$
[\bm{H}_{\gamma}]_{6k-5:6k,m}=\left\{ 
\begin{aligned}
&\bm{h}_{\gamma,k+1}^c,\ \textrm{if }m=2s_{k+1}-1,\\
&\bm{h}_{\gamma,k+1}^s,\ \textrm{if }m=2s_{k+1},\\
&-\bm{h}_{\gamma,k}^c,\ \textrm{if }m=2s_{k}-1,\\
&-\bm{h}_{\gamma,k}^s,\ \textrm{if }m=2s_{k},
\end{aligned}\right.
$$
and $\bm{c}_\gamma=[{\bm{c}}_{\gamma,1}^T,{\bm{c}}_{\gamma,2}^T,\ldots,{\bm{c}}_{\gamma,{K-1}}^T]^T$.

\section{Full Rankness of Matrices $\bm{H}_\rho$ and $\bm{H}_{\vartheta}$}
This section is to show, under Assumption~\ref{assp:regular}, matrices $\bm{H}_\rho$ in~\eqref{eq:opt_rho} and $\bm{H}_{\vartheta}$ in~\eqref{eq:qcqp} are of full rank with probability one. 


First, we note that the structures of $\bm{H}_{\rho}$ and $\bm{H}_{\vartheta}$ are closely related to a binary matrix $\bm{B}\in\{ 0, 1 \}^{K\times M}$ (with $K>M$), where $K$ is the total number of measurements, and $M$ is the number of sensors. Each entry $\bm{B}_{km}=1$ represents that sensor $m$ has a measurement at time instance $k$, and $\bm{B}_{km}=0$ otherwise. Once $\bm{B}$ is of full (column) rank, then it is not hard to show that $\bm{H}_{\rho}$ and $\bm{H}_{\vartheta}$ are also of full (column) rank (this will be discussed later). By Assumption~\ref{assp:regular}, we know that each row of $\bm{B}$ is a one-hot row vector. Then, by applying proper row permutation of $\bm{B}$, we have 
$$
\bm{B} {\xrightarrow{\text{row permutation}}} \begin{bmatrix} \bm{I} \\ * \end{bmatrix}
$$ 
and hence $\bm{B}$ is always of full (column) rank. 
    
Next, we discuss the relation between $\bm{B}$ and matrices $\bm{H}_{\rho}$ and $\bm{H}_{\vartheta}$. We take $\bm{H}_\alpha$ as an example and the analysis is similar for other matrices $\bm{H}_{\rho}$, $\bm{H}_{\beta}$, $\bm{H}_{\gamma}$, and $\bm{H}_{\eta}$. First, we need to expand $\bm{B}$ to a $6K\times 2M$ matrix $\bm{G}$, where $1$ in the $k$-th row of $\bm{B}$ is replaced by $$\bm{H}_{k}=[\bm{h}_{\alpha,k}^c,\bm{h}_{\alpha,k}^s]\in\mathbb{R}^{6\times 2}$$ and $0$ in $\bm{B}$ is replaced by $\bm{0}\in\mathbb{R}^{6\times 2}$. Vectors $\bm{h}_{\alpha,k}^c$ and $\bm{h}_{\alpha,k}^s$ are defined in Appendix~\ref{app:deriv}-C2. According to the definition of $\bm{H}_{k}$, the event that $\bm{H}_{k}$ is not of rank 2 is \emph{with zero probability}, since this event requires that the  measurement noise $\bm{w}_k=(w_k^\rho,w_k^\phi,w_k^\eta)$ (contained in $(\rho_k, \phi_k, \eta_k)$) satisfies the following equation 
$$\bm{h}_{\alpha,k}^c=C\bm{h}_{\alpha,k}^s$$ 
for a constant $C \in \mathbb{R}$. To be specific, denote $$(x,y,z)=R_x(\alpha_{s_{k}})R_y(\beta_{s_k}+\Delta \beta_{s_k}) R_z(\gamma_{s_k}+\Delta \gamma_{s_k})\bm{u}_k.$$ Then, we see that $(x,y,z)$ is uniquely determined by the polar coordinate $(w_k^\rho,w_k^\phi,w_k^\eta)$ and the above equation becomes 
$$
y=-Cz~\text{and}~z = Cy.
$$ It is easy to check that the above linear equation (in terms of $y$ and $z$) has the unique solution $z=y=0.$ However, as $(w_k^\rho,w_k^\phi,w_k^\eta)$ are continuous random variables,  the event that $z=y=0$ happens is of zero probability. 

Finally, let $\bm{G}_{k:l}$ denote the submatrix stacked by the rows of $\bm{G}$ from index $k$ to $l$ ($k\leq l$). Then by Assumption 1 (i.e., each sensor has at least one measurement and $K\geq M+1$), we know that $\bm{G}_{1:6(K-1)}$ is of full (column) rank with probability one. Performing proper elementary row operations for $\bm{G}$, we get $\tilde{\bm{G}}$, where
$\tilde{\bm{G}}_{6(k-1)+1:6k}=\bm{G}_{6(k-1)+1:6k}-\bm{G}_{6k+1:6(k+1)}, \ k=1,2,\ldots,K-1$
and $\tilde{\bm{G}}_{6(K-1)+1:6K}=\bm{G}_{6(K-1)+1:6K}$. 
Due to the above simple row operations and the fact that that $\bm{H}_{k}-\bm{H}_{k'}$ (with $k'\neq k$) is of full (column) rank with probability one (which can be proved using the same argument as in the proof that $\bm{H}_{k}$ is of full column rank), one can show that $\tilde{\bm{G}}_{1:6(K-1)}$ is also of full (column) rank with probability one. Noticing that $\bm{H}_{\alpha}=\tilde{\bm{G}}_{1:6(K-1)},$ we immediately obtain that $\bm{H}_{\alpha}$ is also of full (column) rank with probability one.

Finally,  we use the following toy example to illustrate the above proof. Suppose that there are two sensors and three measurements. Without loss of generality, we assume that sensor $1$ has measurement at time instance $1$ and sensor $2$ has measurement at time instance $2.$ Now there are two different cases where sensor $1$ has measurement at time instance $3$ and sensor $2$ has measurement at time instance $3.$ We first consider the case where sensor $1$ has measurement at time instance $3.$ In this case, we have $$\bm{G}=\left[\begin{array}{cc}
\bm{H}_{1}& \bm{0}  \\
\bm{0}& \bm{H}_{2}  \\
\bm{H}_{3}& \bm{0}  
\end{array}\right]
,~\tilde{\bm{G}}=\left[\begin{array}{ccc}
\bm{H}_{1}& -\bm{H}_{2}  \\
-\bm{H}_{3}& \bm{H}_{2}  \\
\bm{H}_{3}& \bm{0}
\end{array}\right]
$$ 
and 
$$\bm{H}_{\alpha}=\left[\begin{array}{ccc}
\bm{H}_{1}& -\bm{H}_{2}  \\
-\bm{H}_{3}& \bm{H}_{2} 
\end{array}\right].$$
One can show that $\bm{H}_{1}-\bm{H}_{3}$ is of full (column) rank with probability one and hence $\bm{H}_{\alpha}$ in the above is also of full (column) rank with probability one. In the second case where sensor $2$ has measurement at time instance $3$, we have $$\bm{G}=\left[\begin{array}{cc}
\bm{H}_{1}& \bm{0}  \\
\bm{0}& \bm{H}_{2}  \\
 \bm{0} & \bm{H}_{3}  
\end{array}\right]
,~\tilde{\bm{G}}=\left[\begin{array}{ccc}
\bm{H}_{1}& -\bm{H}_{2}  \\
\bm{0} & \bm{H}_{2} -\bm{H}_{3}  \\
\bm{0}& \bm{H}_{3}
\end{array}\right]$$
and 
$$\bm{H}_{\alpha}=\left[\begin{array}{ccc}
\bm{H}_{1}& -\bm{H}_{2}  \\
\bm{0} & \bm{H}_{2} -\bm{H}_{3}
\end{array}\right].$$ 
Again, $\bm{H}_{2} -\bm{H}_{3}$ and $\bm{H}_{\alpha}$ in the above are of full (column) rank with probability one.
\end{document}